\documentclass[preprint,5p,times,twocolumn]{elsarticle}
\journal{...}

\usepackage{soul}
\usepackage{comment}
\usepackage{forest}
\usepackage{xcolor}
\usepackage{booktabs}
\usepackage{tabu}
\usepackage{multirow}
\usepackage{url}
\usepackage{hyperref}
\hypersetup{
  colorlinks=true,
  linkcolor=black,
  urlcolor=black,
  citecolor=black,
}
\usepackage{balance}
\usepackage{amsfonts}
\usepackage{amssymb}
\usepackage{balance}

\usepackage{tikz}

\usepackage{glossaries}
\newacronym{hhl}{HHL}{Harrow–Hassidim–Lloyd}
\newacronym{ai}{AI}{Artificial Intelligence}
\newacronym{ml}{ML}{Machine Learning}
\newacronym{pqc}{PQC}{Post-Quantum Cryptography}
\newacronym{qkd}{QKD}{Quantum Key Distribution}
\newacronym{qsdc}{QSDC}{Quantum Secure Direct Communication}
\newacronym{qml}{QML}{Quantum ML}
\newacronym{qpp}{QPP}{Quantum Permutation Pad}
\newacronym{qsvm}{QSVM}{Quantum Support Vector Machine}
\newacronym{xqml}{XQML}{eXplainable Quantum ML}
\newacronym{fpga}{FPGA}{Field Programmable Gate Arrays}
\newacronym{svm}{SVM}{Support Vector Machine}

\usepackage{amsthm}
\newtheorem{theorem}{Theorem}[section]

\newtheorem{lemma}[theorem]{Lemma}
\newtheorem{definition}{Definition}[section]
\newtheorem{remark}{Remark}[section]

\usepackage{braket}
\usepackage{mdframed}
\usepackage{lipsum,adjustbox}

\usepackage{tikz}
\usetikzlibrary{quantikz2}
\usepackage{lipsum,adjustbox}
\usetikzlibrary{calc,positioning}

\usepackage{forest}

\usepackage{subcaption}

\newif\ifcomments

\newcommand{\update}[1]{
    \ifcomments
        \textcolor{black}{#1}
    \fi
}

\commentstrue

\newcounter{reviewer}
\newcounter{point}[reviewer]
\begin{document}

\begin{frontmatter}

\title{
Identifying vulnerable nodes and detecting malicious entanglement\\ patterns to handle \emph{st}-connectivity attacks in quantum networks} 

\author[1]{Iain Burge}
\affiliation[1]{organization={SAMOVAR, Télécom Sudparis, Institut Polytechnique de Paris},
            city={91120 Palaiseau},
            country={France}}
\author[2]{Michel Barbeau}
\affiliation[2]{organization={School of Computer Science, Carleton University},
            city={K1S 5B6 Ottawa},
            country={Canada}}
\author[1]{Joaquin Garcia-Alfaro}

\begin{abstract}
\update{Several problems in distributed system security naturally map to graphs. The concept of centrality assesses the importance of nodes in a graph. It is used in various applications. Cooperative game theory has also been used to create nuanced and flexible notions of node centrality. However, the approach is often computationally complex to implement in classical settings. Our first contribution describes a quantum approach to approximating the importance of quantum nodes that support a target connection in a quantum network. We detail a method for quickly identifying high-importance nodes that adversaries can target.  The approximation method relies on quantum subroutines for evaluating \emph{st}-connectivity, approximating Shapley values, and finding the maximum of a list. We consider a malicious actor targeting a subset of nodes to disrupt the system functionality.  Our method identifies the nodes that are most important in keeping nodes $s$ and $t$ connected.  Once we have identified high-importance nodes, we require methods to identify when those nodes are compromised. Our second contribution describes how Quantum Support Vector Machine (QSVM) classifiers can be used to detect malicious behavior in quantum networks.  In particular, we describe the detection of entanglement attacks in quantum repeaters. We show that our initial assessment approach can be complemented by QSVM classifiers to identify and alert when anomalous situations related to malicious manipulation of entanglement swapping occur. Finally, we explore the potential complexity benefits of our quantum approach compared with classical and probabilistic methods. We also release all the simulation code and artifacts associated to the work, to foster reproducibility and further research on the topic.\\}
\end{abstract}


\end{frontmatter}

\section{Introduction}
\label{sec:introduction}

\medskip

\noindent \update{With recent promising results in quantum computing for combinatorial optimization problems, quantum-enhanced information networks are a promising evolution of classical distributed systems, in which the use of quantum technologies is expected to foster significant new paradigms~\cite{noirie2023,gyongyosi2025networked,chung2026interqnet}.  This includes the development of quantum sensor networks and the enhancement of quantum safe communication technologies~\cite{cao2022evolution}.  Yet, these new environments must face traditional security problems, including defense and resilience.}

In the realm of graph analytics, node centrality metrics quantify properties such as a node's utility, whether it is critical to keeping the graph connected, or whether it is vulnerable to attacks. These metrics help to determine whether a network is secure and resilient. They can guide structural changes to improve these properties.\update{Traditional node centrality metrics assess individual nodes; however, some properties cannot be easily measured without considering node coalitions.}

\update{The first contribution of this article builds upon the flexible notion of game-theoretic node centrality measures. We describe a node centrality metric based on connecting two critical nodes.  We use the metric to address the following two properties relevant to distributed systems security: resilience and the degree of remediation. The former refers to the ability of a communication network to maintain functionality and carry out its mission, even in the face of adversarial events. An adversarial event can either occur naturally or result from deliberate actions. The latter can be used to quantify the capacity to provide restoration and mitigation after an attack on the system. In this respect, we aim at addressing the aforementioned properties in the presence of adversaries in a distributed system perpetrating a given type of attack, namely, the \emph{st}-connectivity attack. We build a methodological solution to assess a node's importance. Quantifying the importance of nodes can guide modifications to the network topology to improve resilience. We also explore the advantage of a quantum version of our solution, compared to a baseline classical computing solution.}

\update{Once the potential victims of an attack are identified, the second contribution of this article builds upon \gls*{qsvm} classifiers to confirm the presence of an ongoing attack.  For instance, we consider the case of a malicious quantum repeater perpetrating malicious entanglement swapping}~\cite{satoh2018network,satoh2021attacking}.
\update{The goal of the attack is to alter quantum-safe conditions between end nodes. We provide a complexity analysis for this second contribution as well as experimental results. Implementation code and artifacts are also released in a companion repository~\cite{github}.}

\medskip

\noindent \update{\textbf{Article organization\footnote{This is a revised and extended version of a paper that appeared in the proceedings of the 40th IFIP International Conference on ICT Systems Security and Privacy Protection (IFIP SEC 2025), Part II, Pages 234--248, Maribor, Slovenia, May 21-23, 2025~\cite{burge@ifipsec2024}.}  ---}  Section~\ref{sec:motivation} further elaborates our motivation and threat models. Section~\ref{sec:preliminaries} presents some necessary preliminaries. Section \ref{sec:approach-1} describes our first contribution, on approximating the importance of nodes that maintain a target connection.  Section \ref{sec:second-contribution} \update{outlines} our second contribution on extending the assessment approach together with \gls*{qsvm} classifiers to identify malicious events associated with our threat model. Section \ref{sec:experiments} reviews experimental results and complexity evaluation of our two contributions. Section~\ref{sec:merged-RW} discusses related work. Section~\ref{sec:conc} provides conclusions and perspectives for future work.}

\section{Motivation}
\label{sec:motivation}

\begin{figure}[!t]
\centering
\begin{subfigure}[t]{.98\columnwidth}
    \includegraphics[width=\columnwidth]{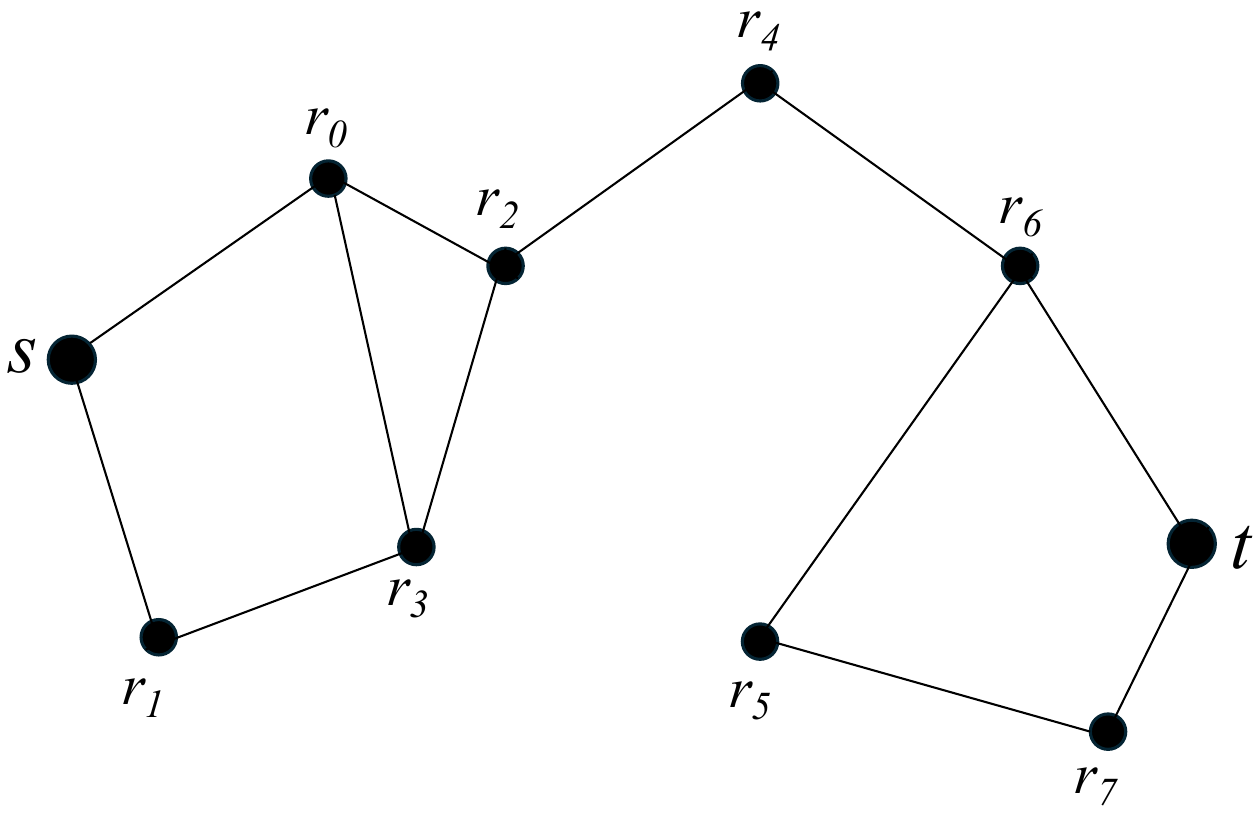}
    \caption{Quantum network protocol stack from~\cite{noirie2024}. }
    \label{subfig:basicNetwork}
\end{subfigure}
\begin{subfigure}[t]{.98\columnwidth}
    \includegraphics[width=\columnwidth]{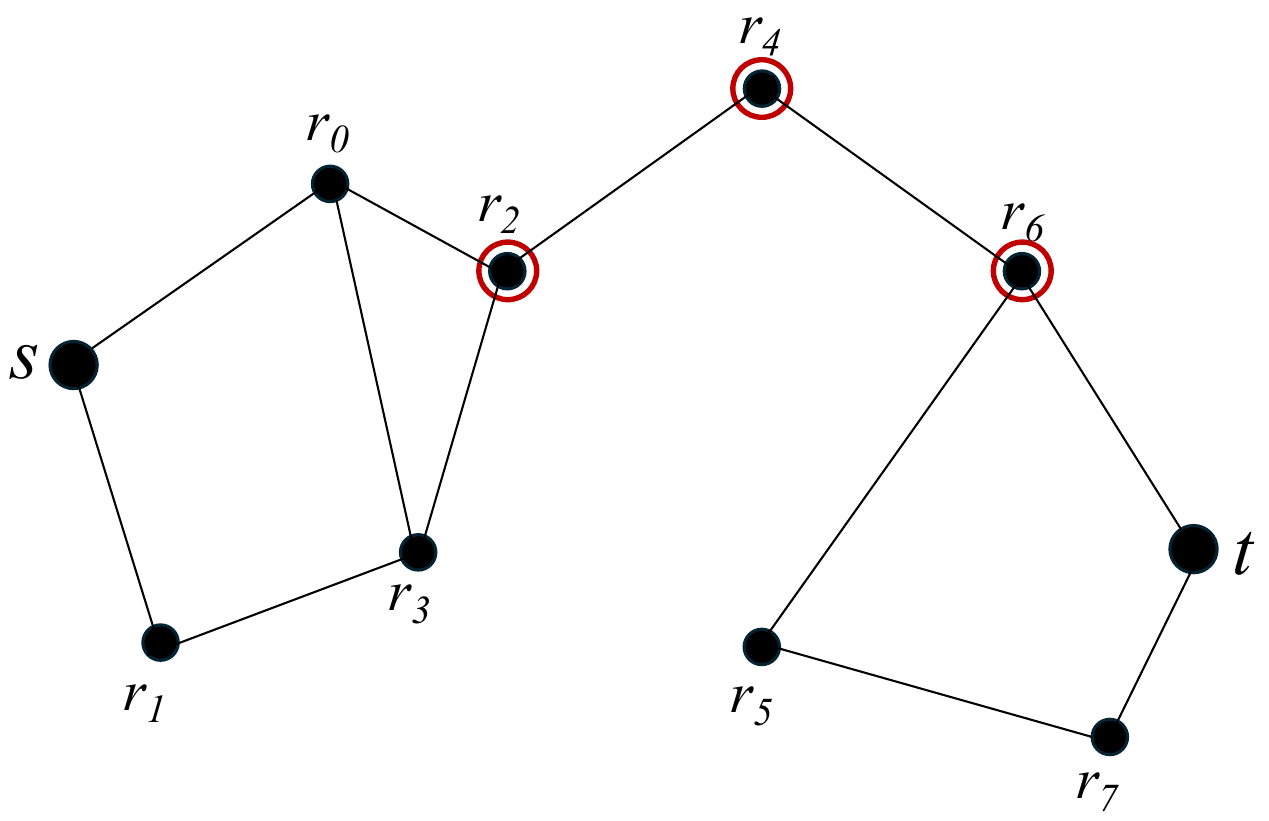}
    \caption{Identification of critical nodes (denoted by red circles).\label{subfig:identification}}
\end{subfigure}
 \begin{subfigure}[t]{.98\columnwidth}
     \includegraphics[width=\columnwidth]{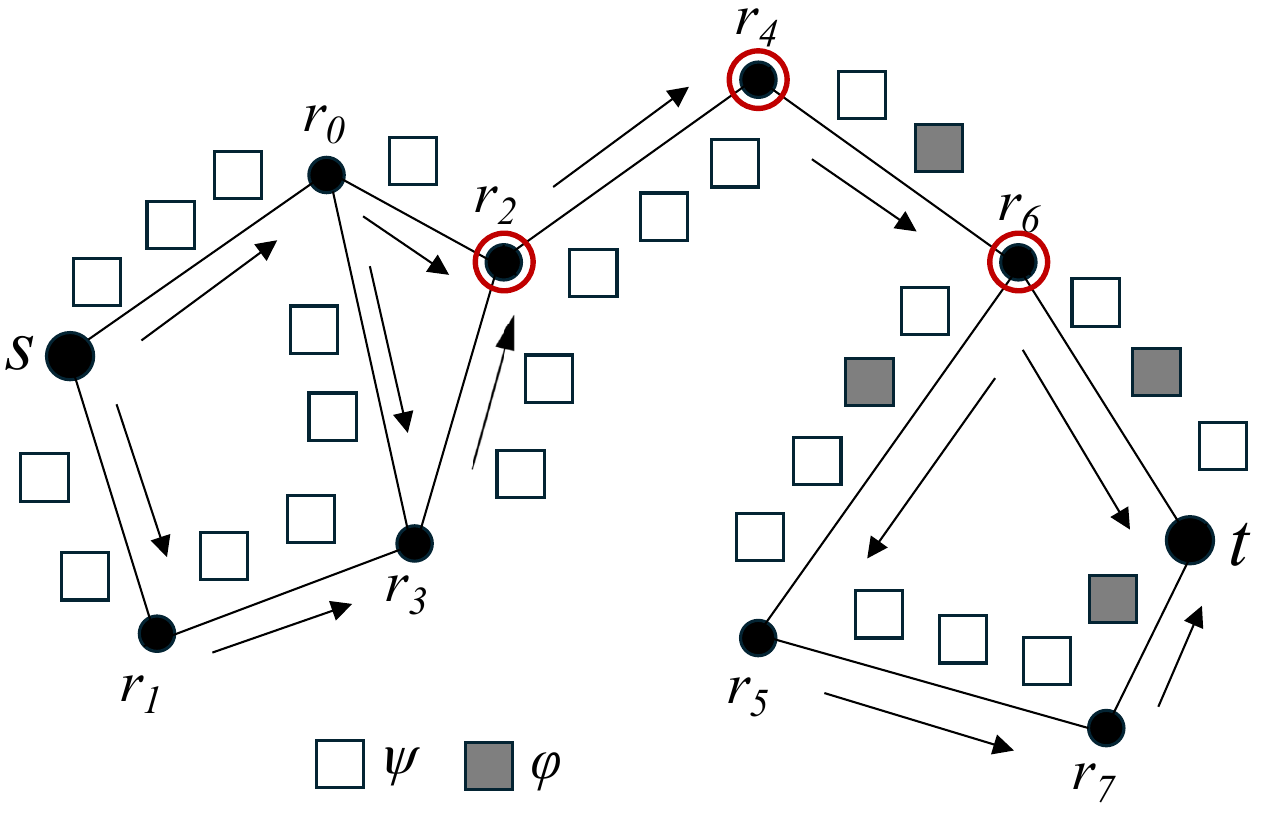}
     \caption{State manipulation. White squares represent legitimate entangled states in transit from $s$ to $t$ (denoted by arrows). Grey squares represent malicious state modification. \label{subfig:detection}}
 \end{subfigure}
\caption{Motivation scenario, assuming the quantum network protocol stack from~\cite{noirie2024,chung2026interqnet}. Small black circles ($r_0$, $r_1$, $\ldots$, $r_7$) in (a) represent quantum repeaters (i.e., quantum nodes of degree two, conducting only entanglement swapping) and routers (i.e., quantum nodes of degree greater than or equal to three, conducting both entanglement swapping and routing). Large black circles ($s$ and $t$) represent quantum computing nodes, i.e., running quantum applications, in addition to routing. The scenario depicts the identification of critical nodes in (b), as well as pattern detection of malicious entanglement situations in (c). Nodes with red circles depict an identified coalition of critical nodes, followed by its augmented scenario, in which a classifier monitors the system to detect malicious actions within the coalition of critical nodes (e.g., state manipulation).} 
\label{fig:motivation}
\end{figure}

\subsection{Threat and adversary models}
\label{sec:motivation-part1}
\noindent Our threat model extends the work by Satoh \emph{et al.} \cite{satoh2021attacking,lloyd2026risk}, which des\-cribes representative attacks in which legitimate repeaters can be framed by malicious nodes with higher computational capabilities. This strategy can be used to compromise the topology of a network, in the worst case, partitioning it entirely.
As a result, it is natural to place greater emphasis on individual nodes, especially those involved in many potential routes.

We now describe the adversary model.
Suppose that there is critical quantum communication between nodes $s$ and $t$ (cf. Figure~\ref{subfig:basicNetwork}), which represent quantum computing nodes running quantum applications (e.g., entangled versions of \gls*{qkd}~\cite{cao2022evolution} and \gls*{qsdc}~\cite{pan2024evolution}), in addition to routing and entanglement swapping. We assume an adversary with a limited budget. As a result, the adversary can only compromise a limited number of nodes, for instance, quantum repeaters (i.e., quantum nodes of degree two, conducting only entanglement swapping). Adjacent to the adversary, we assume quantum routers (i.e., quantum nodes of degree greater than or equal to three, capable of conducting both entanglement swapping and routing). The adversary, to maximize the effect of the attack, chooses a node required in many paths between $s$ and $t$ and performs malicious entanglement actions, impacting system confidentiality, integrity, and availability \cite{satoh2021attacking}.

In particular, we consider a malicious node $z$ connecting two uncompromised nodes $x$ and $y$, where $x$ is trying to send an entangled pair of qubits $A,B$ to $y$.
The compromised node $z$ takes $A,B$ as input, but redirects $B$, instead sending $A$ and a new qubit $C$ to $y$, where $A$ and $C$ are unentangled.
To address this adversary, we propose methods to identify likely targets and to monitor and assess the entanglement between a pair of particles. More details are provided in the next section.

\subsection{Identifying targets and detecting state manipulation}
\label{sec:motivation-part1}

\noindent We assume a quantum distributed system offering, e.g., quantum key expansion, entanglement swapping, and error mitigation services~\cite{noirie2024,cao2022evolution}, in which an adversary aims to disrupt service connectivity from node $s$ to node $t$. By assuming a classical abstraction of the problem and by focusing only on the information-gathering stage of the attack, we aim to anticipate ways for the adversary to identify the best strategies to disconnect 
$t$ from $s$ (i.e., we assume that the adversary can successfully sabotage the services in those intermediate nodes from $s$ to $t$, hence avoiding any possible functionality between both nodes).
Our goal is to identify the most important nodes in keeping node $t$ accessing the services of $s$, allowing us to increase the resilience of the network (cf. Figure~\ref{subfig:identification}). 
We accomplish this by using the game-theoretic concept of node centrality as a metric to quantify the remediation degree associated with the attack scenario.

Then, we assume an adversary perpetrating malicious entanglement against the swapping services of a quantum network~\cite{van2013designing} (cf. Figure~\ref{subfig:detection}). We aim at detecting such attack patterns. We extend previous work in~\cite{burge@ifipsec2024,burge2023quantum,qce2020}, in which we already described a quantum approach using the Shapley theory, to approximating the importance of nodes that maintain a target connection, as well as how to quickly identify high-importance nodes and repeater placement strategies. Therefore, and after identifying critical nodes conducting entanglement swapping, we propose to go a step further and leverage proactive defense triggered by pattern detection of malicious activities.

\section{Preliminaries}
\label{sec:preliminaries}

\noindent Before moving forward, We start with some background concepts on which our work is founded.

\begin{definition}[Graph]\label{def:graph}
    Define a graph $H=(N,E)$ to be a pair of the set of network nodes $N=\{0,1,\dots,m-1\}$ and set of edges $(u,w)$ in $E$, with $u,w\in N$ and $u\neq w$. 
\end{definition}

\begin{remark}[Graph representation]\label{rem:graphRepr}
    Let us index each node by some integer in $\mathbb{Z}_{|N|}$\footnote{Given an integer $m$, $\mathbb{Z}_{m}$ is the set $\{0,1,\dots,m\}$.}.
    Each edge is indexed by an integer in $\mathbb{Z}_{|N|(|N|-1)}$, that is mapped to the set of pairs $\{(u,w): u,w\in \mathbb{Z}_{|N|}, u\neq w\}$, with a bijection.
    We write the index of edge $(u,w)$ as $uw$.
    We may represent the adjacency matrix of the graph with a binary string $x\in\{0,1\}^{|N|(|N|-1)}$, where $x_{uw}$ is one if $(u,w)\in E$, otherwise $x_{uw}$ is zero.
\end{remark}
\begin{definition}[Cooperative games on a network graph \cite{tarkowski2017game}]\label{def:coopGame}
    We define a cooperative game on graph \hbox{$H=(N,E)$} to be the pair $G_H = (F,V)$, where $F\subseteq N$ and $V$ is a valuation function from the subsets of $F$ to the reals, i.e., $V:\mathcal{P}(F) \rightarrow \mathbb{R}$.
    With the restriction that $V(\emptyset)=0$.
\end{definition}
This definition allows us to treat the nodes in $F$ as players in a game.
Given a subset of nodes $R\subseteq F$, we can treat it as a binary graph coloring where the colors correspond to the inclusion (or exclusion) of the node in $R$.
$V(R)$ represents the value of that particular graph coloring. Though it is useful to have a value for coalitions of nodes, or their colorings, the number of combinations grows exponentially with respect to graph size.
Thus, it is useful to have a metric that can condense this vast amount of information into a utility for each node.
We adapt the Shapley value concept to our current situation.
\begin{definition}[Node Shapley value \cite{shapley1952value}]\label{def:shapley}
    Given a game $G_H=(F,V)$ on graph \hbox{$H=(N,E)$}, with $F\subseteq N$.
    The $i$th node's Shapley value $\Phi_i$ is,
    \begin{equation*}
        \sum_{R\subseteq F\setminus \{i\}} 
        \gamma\left(|F\setminus \{i\}|,|R|\right) \cdot
        \left(V(R\cup \{i\}) - V(R) \right)
    \end{equation*}
    where $\gamma(n,m) = \left( {n \choose m} (n+1)\right)^{-1}$.
\end{definition}

The Shapley value of node $i$ can be understood intuitively.
It is a weighted average of the increase in value when node $i$ is added to subsets that exclude $i$.
We proceed with a narrow concept of graph coloring. If node $a\in F$ is in $R$, it is considered \emph{enabled}, otherwise, if $a$ is not in $R$, $a$ is considered \emph{disabled}.

\begin{definition}[Sub-graph $H_Q$]\label{def:subgraph}
    We define the sub-graph $H_Q=(Q,E_Q)$ of the graph $H=(N,E)$, such that $Q\subseteq N$. 
    $E_Q\subseteq E$ is the subset of all edges $(a,b)\in E$ where $a,b\in Q$.
\end{definition}

In the context of node centrality, we consider the value function $V(R)$ that indicates whether $H_R$ maintains a particular property.

\subsection{The \textit{st}-connectivity attack}

\label{sec:stConnectivity}
\begin{definition}[\emph{st}-connectivity]\label{def:stConnect}
    Consider a graph $H=(N,E)$, with nodes $s,t\in N$.
    The graph $H$ is $st$-connected if there exists a path from node $s$ to node $t$.
    Formally, $H$ is $st$-connected if there exists a sequence of nodes $s=u_0, u_1, u_2, \dots, u_{r-1}, u_m=t$ such that \hbox{$(u_k, u_{k+1})\in E$} for $k\in\{0,\dots, m-1\}$.
    We define the value function $V_{st}:\mathcal{P}(F) \rightarrow \mathbb{R}$,
    \begin{equation*}
        V_{st}(R) = \begin{cases}
            1 & \text{if $H_{R\cup \{s,t\}}$ is $st$-connected,}\\
            0 & \text{otherwise,}
        \end{cases}
    \end{equation*}
    where $R\subseteq F = N\setminus \{s, t\}$, and $H_{R}$ is described in Definition~\ref{def:subgraph}.
\end{definition}

In the context of our scenario, the adversary aims to remove \emph{st}-connectivity (source-target-connectivity).
The value function returns $0$ when the set of enabled nodes $H_{R}$ is no longer able to keep the target connected to the source, and $1$ when it maintains that property.
Hence, the Shapley values (Definition~\ref{def:shapley}) of each node reflect how critical it is to maintain that connection.
A high Shapley value means that the node is a valuable target, while a low Shapley value means that the node is not of interest.

\begin{definition}[\emph{st}-connectivity attack]\label{def:attack}
    Given a graph $H=(N,E)$, an $st$-connectivity attack is a malicious action perpetrated by an adversary. 
    The adversary can turn off a subset of nodes $Q\subseteq F = N\setminus \{s,t\}$. 
    The adversary's goal is to transform the graph $H$ into a sub-graph $H_{N \setminus Q}$ that is not $st$-connected.
    Equivalently, the adversary's goal is to minimize $V_{st}(F\setminus Q)$.
\end{definition}

\section{Quantum assessment of critical nodes}
\label{sec:approach-1}
%=========================================%
%=========================================%

\noindent We present  our quantum algorithm for \emph{st}-connectivity assessment. 
To begin, we define a simplified version of span programs, detailed in \cite{belovs2012span, cade2016time}.

\begin{definition}[Span program decision problem] 
    \label{def:spanProgram}
    A \emph{span program} $P\left(\ket \tau, \mathcal{W}, x\right)$ takes as input a dimension $d$ unit \emph{target vector} $\ket\tau\in \mathbb{C}^d$, 
    a set of \emph{input vectors} $\mathcal{W}=\{\ket{\mu_{k,0}}:k\in \mathbb{Z}_r\}\cup\{\ket{\mu_{k,1}} : k\in\mathbb{Z}_r\} \subset \mathbb{C}^d$, and a length $r$ binary vector selection string $x=x_{r-1}\cdots x_0\in\{0,1\}^r$.
    Note that the \emph{input vectors} need not be unit length.
    The binary string $x$ determines the vectors of $\mathcal{W}$ to be used.
    Formally, given $x$, the \emph{available vectors} are $A=\{\ket{\mu_{k,x_k}}: k\in\mathbb{Z}_r\}\subset \mathcal{W}$.
    The span program $P$ outputs $1$ if the target $\ket \tau$ is in the span of the available vectors in $\texttt{Span}(A)$.
    Equivalently, $P$ outputs $1$ if $\ket\tau$ can be written as a linear combination of available vectors $A$,
    \begin{equation*}
        \ket \tau = \sum_{k=0}^{r-1} c_k \ket{\mu_{k,x_k}},\quad c_k\in \mathbb{C}.
    \end{equation*}
    Otherwise, the program returns $0$.
\end{definition}

We now reformulate the problem of \emph{st}-connectivity as a span program decision problem~\cite{belovs2012span}.
\begin{theorem}[Span program for \emph{st}-connectivity]\label{thm:spanSTconnectivity}
    Consider graph $H=(N,E)$, as in Definition~\ref{def:graph}.
    We detail a span program that determines, given $s,t\in N$, if $H$ is $st$-connected.
    If $v$ is a node in $N$, we define the basis vector $\ket v$ to correspond to the node.
    Define $P\left(\ket \tau, \mathcal W, x\right)$, where $\ket \tau\in \mathbb C^{|N|}$, 
    \begin{equation*}
        \mathcal W = \left\{ \ket{\mu_{uw,0}}:(u,w)\in E\right\}\cup \left\{\ket{\mu_{uw,1}} : (u,w)\in E\right\},
    \end{equation*}
    and $x$ is the binary string representation of the adjacency matrix for $H$ (Remark~\ref{rem:graphRepr}).
    The value of $x_{uw}$ is $1$ if $(u,w)\in E$; otherwise, $x_{uw}$ is $0$.
    The target vector is,
    \begin{equation*}
        \ket \tau = \frac{\ket t - \ket s}{\sqrt 2}, \quad s,t\in N.
    \end{equation*}
    The input vectors are $\ket{\mu_{uw,0}}=0$, and, $\ket{\mu_{uw,1}} = (\ket u - \ket w)/{\sqrt2}$, 
    for all $u,w \in N$, and edge indices $uw\in \mathbb{Z}_{|N|(|N|-1)}$.
    Thus, our available vector span is,
    \begin{equation*}
        \texttt{Span}(A) = \texttt{Span}\left\{\frac{\ket u-\ket w}{\sqrt{2}} : x_{uw}=1,uw\in \mathbb{Z}_{|N|(|N|-1)}\right\}.
    \end{equation*}
    If the span program outputs $1$, then
    $H$ is $st$-connected; otherwise, $H$ is not $st$-connected.
\end{theorem}
\begin{proof}
    Suppose $H=(N,E)$ is \emph{st}-connected, then there exists a sequence of nodes $s=u_0,\dots,u_{r-1},u_m=t$, such that $(u_k,u_{k+1})\in E$, $k\in \{ 0,\dots, m-1\}$.
    As a result, for our span program $P(\ket \tau, \mathcal W, x)$, 
    the set of available vectors $A$ includes every
    \begin{equation*}
        \frac{\ket{u_{k+1}} - \ket{u_k}}{\sqrt{2}},  \quad with~k\in \mathbb{Z}_r.
    \end{equation*}
    We have,
    \begin{equation*}
        \ket\tau = \sum_{k=0}^{r-1} \frac{\ket{u_{k+1}} - \ket{u_k}}{\sqrt2},
    \end{equation*}
    since the right-hand side is a telescoping sequence.
    As a result, the span program accepts the input as expected.
    A proof validating that the span program rejects $H$ when it is not \emph{st}-connected exists \cite{cade2016time}.
\end{proof}

\begin{theorem}[Quantum \emph{st}-connectivity algorithm \cite{belovs2012span,cade2016time}]\label{thm:quantumSTconnectivity}
    There exists a quantum algorithm to decide whether a graph $H=(N,E)$, with nodes $s,t\in N$, is $st$-connected.
    The algorithm uses $\mathcal{O}\left(\log|N|\right)$ space, takes $\mathcal{O}\left( |N|^\frac{3}{2} \right)$ queries to the adjacency matrix up to polylogarithmic factors.
    The routine succeeds with probability at least $9/10$.
    The best possible classical algorithm takes at least $\Omega\left(|N|^2\right)$ time.
    
    Formally, we have a unitary quantum transformation $U_{st}$
    which acts on an auxiliary register of $\mathcal O(\log|N|)$ qubits \texttt{aux}~and an output register of one qubit \texttt{out}.
    Performing the algorithm and tracing out the auxiliary register results in,
    \begin{equation*}
        \texttt{tr}_\texttt{aux} \left(U_{st}\ket{0}^{\otimes\mathcal{O}(\log |N|)}_\texttt{aux}\ket{0}_\texttt{out}\right) = 
        \left( 
            (1-p)\ket{\neg y}\bra{\neg y} + p\ket{y}\bra{y}
        \right)_\texttt{out}
    \end{equation*}
    where $y$ is one if $H$ is $st$-connected and zero otherwise, and $p$ is in range $[9/10,1)$.
    Measuring the output bit returns the correct output with probability $p$.
\end{theorem}
\begin{proof}
    We proceed with a rough sketch of the algorithm.
    A full algorithm and proof are provided in~\cite{cade2016time}.
    The algorithm is based on the span program for \emph{st}-connectivity.
    We perform phase estimation on the unitary matrix $U=(2\Lambda-I)(2\Pi_x-I)$ with the input vector $\ket{0}$ using precision $\mathcal{O}(|N|^{3/2})$.
    Thus $U$ is queried $\mathcal{O}(|N|^{3/2})$ times.
    If the phase estimation outputs zero, the algorithm claims that the graph $H$ is \emph{st}-connected and outputs $1$.
    Otherwise, if the phase estimation outputs a non-zero answer, the algorithm claims that $H$ is not \emph{st}-connected, and outputs $0$.
    It is correct with probability $9/10$.
    We assume, for the sake of simplicity, that $(s,t)\notin E$, this can be checked in $\mathcal{O}(|N|)$ time.
    We also give edge $(s,t)$ the index $st=0$.

    $U$ is the product of two reflections, a reflection about $\Lambda$, and a reflection about $\Pi_x$.
    $\Lambda$ represents a projection onto the kernel of,
    \begin{equation*}
        \tilde M = \mathcal{O}\left(\frac{1}{\sqrt{|N|}}\right)\ket{\tau}\bra{0}+\sum_{uw\in \mathbb Z_{|N|(|N|-1)}\setminus\{0\}}\ket{\mu_{uw,1}}\bra{uw}.
    \end{equation*}
    $\tilde M$ represents a transformation from the indices of edges to their respective vectors in the span program for \emph{st}-connectivity. 
    The reflection, $(2\Lambda-I)$, is implemented using a Szegedy-type quantum walk~\cite{cade2016time, szegedy2004quantum}. 
    The walk is implemented in logarithmic space and time with respect to $|N|$, and is input independent.
    $\Pi_x$ is the projection onto available vector indices and onto the target vector index,
    \begin{equation}\label{eq:Pi_x}
        \Pi_x = \ket{0}\bra{0} + \sum_{(u,w)\in E} \ket{uw}\bra{uw}.
    \end{equation}
    Thus, $(2\Pi_x-I)$ represents a reflection where all the indices of unavailable edges are negated.
    This reflection can be performed with a single query to the adjacency matrix.
    
    Intuitively, the quantum phase estimation extracts the spectral qualities of $U$.
    The reflections $(2\Lambda-I)$ and $(2\Pi_x-I)$ are constructed such that the spectral qualities of $U$ correspond to whether $\ket\tau$ is linearly independent of the available vectors.
\end{proof}

\begin{remark}[Span program for \emph{st}-connectivity node centrality]\label{rem:spanSTCentrality}
    Consider the graph $H=(N,E)$.
    Suppose we wish to ascertain the $st$-connectivity of a sub-graph $H_R=(R,E_R)$, $R\subseteq N$.
    Equivalently, we wish to compute $V_{st}(R)$ defined in~\ref{def:stConnect}.
    We proceed similarly as in Theorem~\ref{thm:spanSTconnectivity}.
    Define the span program $P(\ket \tau,\mathcal W,x^R)$, where $\ket \tau$ and $\mathcal W$ are described in Theorem~\ref{thm:spanSTconnectivity}.
    Let $x^R_{uw}$ be one if $uw\in E_R$, otherwise $x^R_{uw}$ is zero.
    Equivalently, we can define $x^R_{uw}$ to equal one if and only if $x_{uw}$ is one and nodes $u,w\in R$.
\end{remark}

\begin{definition} [Majority vote]
    We define the majority function $\texttt{MAJ}:\{0,1\}^n\rightarrow{0,1}$, where $n$ is odd, as,
    \begin{equation*}
        \texttt{MAJ}(z) = \begin{cases}
            1 & \text{if $\sum_{k=0}^{n}z_k > n/2$},\\
            0 & \text{otherwise.}
        \end{cases}
    \end{equation*}
    where $z=z_{n-1}\cdots z_0\in\{0,1\}^n$.
    We also describe the quantum version of this function $U_\texttt{MAJ}$ which operates on an $n$-qubit register \texttt{in} and a one-qubit register \texttt{maj},
    \begin{equation*}
        U_\texttt{MAJ}\ket{z}_\texttt{in}\ket{0}_\texttt{maj} = \ket{z}_\texttt{in}\ket{\texttt{MAJ}(z)}_\texttt{maj}.
    \end{equation*}
\end{definition}

\begin{lemma} [Majority vote powering]\label{lem:majVotePower}
    Suppose we have a quantum algorithm $U$ which outputs a binary value with fixed success probability $p>0.5$.
    Let the correct value be $y\in\{0,1\}$.
    We can augment the probability of success by repeatedly performing the algorithm and taking the majority output.
    In particular, suppose our repeated quantum subroutine gave an $n$-qubit output of,
     \begin{equation*}
        \left( 
            (1-p)\ket{\neg y}\bra{\neg y} + p\ket{y}\bra{y}
        \right)^{\otimes n}.
    \end{equation*}
    Then, adding an extra qubit in the form of a $\texttt{maj}$ register, the majority vote unitary $U_\texttt{MAJ}$ can be applied. 
    Given a desired final failure probability bound $\delta$, the \texttt{maj} register stores the correct answer with probability $1-\delta$ if $n$ is of order $\mathcal{O}\left(\log\delta^{-1} \right)$.
    In other words, we have failure chance $\delta$ given $\mathcal O(\log \delta^{-1})$ applications of the $U$ algorithm.
\end{lemma}
\begin{proof}
    Suppose we perform our quantum algorithm $n$ times, where $n\geq3$ is odd.
    This outputs a list of $n$ bits.
    The probability that $k$ bits are correct is,
    \begin{equation}\label{eq:successCount}
        {n \choose k} p^k (1-p)^{n-k}.
    \end{equation}
    The threshold for a majority is $t=(n-1)/2$.
    Hence, the probability that the majority fails is $\sum_{k=0}^{t} {n \choose k}p^k(1-p)^{n-k}$.
    In Equation~\eqref{eq:successCount}, for $k\in[0,(n-1)/2]$, the probability is increasing with respect to $k$.
    Thus, the probability of majority failure is bounded by,
    \begin{equation}\label{eq:failChanceBound}
        t {n\choose t} p^t (1-p)^{n-t}.
    \end{equation}
    By an improved version of Stirling's formula~\cite{robbins1955remark},
    \begin{equation*}
        {n\choose t}< \sqrt{\frac{n}{2\pi t(n-t)}} \frac{n^n}{t^t(n-t)^{n-t}} < \sqrt{\frac{2}{\pi n}} 2^n,
    \end{equation*}
    where the latter inequality is the result of replacing $t$ with $n/2$.
    Plugging the inequality into Equation~\eqref{eq:failChanceBound} and once again replacing $t$ with $n/2$ yields the new bound $\sqrt{\frac{n}{2\pi}} 2^n (p(1-p))^{n/2}$.
    So long as $\sqrt{p(1-p)}<1/2$, which holds for $p>0.5$, the probability of majority failure shrinks exponentially with respect to $n$.
\end{proof}

\subsection{Quantum algorithm for Shapley value approximation}

\noindent The quantum algorithm for Shapley value approximation takes an approach inspired by classical random sampling \cite{castro2009polynomial}.
Each subset of nodes is given a probability amplitude proportional to their $\gamma$ coefficient in the Shapley equation (Definition~\ref{def:shapley}).
Classically, we would randomly sample from the distribution of node subsets, and record how much our target node increases the value of the subset.
After many samples, we take the average increase in value and use it as an approximation.
By Chebyshev's inequality the number of samples required scales quandraticaly with respect to desired error.
The quantum approach can provide a quadratic improvement.

\begin{theorem} [Quantum algorithm for Shapley value approximation \cite{burge2023quantum, burge2024shapley}] \label{thm:quantumShap}
    Take cooperative game on graph $H=(N,E)$ to be the pair $G_H=(F,V)$ where $F\subseteq N$ and $V$ is the value function. 
    Suppose we have a quantum implementation of $V$, $U_V$, and that we wish to find the Shapley value $\Phi_i$ of node $i$.
    Then, given a fixed desired probability for success, there exists a quantum algorithm that produces approximation $\tilde\Phi_i$ in,
    \begin{equation*}
        \mathcal{O}\left(
        \frac{\sqrt {(V_{\max}-V_{\min})(\Phi_i-V_{\min})}}{\epsilon}
        \right),
    \end{equation*}
    queries to the value funciton $U_V$,
    where $V_{\max}, V_{\min}$ are respectively an upper and lower bound for the value function $V$, and the desired error bound is $\epsilon\geq |\Phi_i-\tilde \Phi_i|$.
\end{theorem}
\begin{proof}
    We now give a sketch of the algorithm; a complete proof and error analysis is provided in \cite{burge2024shapley}.
    We can uniquely encode a sub-graph $H_R$, $R\subseteq F$, as a binary string of the form:
        $b^R=b_0^R b_1^R \cdots b_{|F|-1}^R \in \{0,1\}^{|F|}$,
    where $b_j^R=1$ if $k\in R$ else $b_j^R=0$.
    We define quantum implementation $U_V$ of $V$ as,
    \begin{scriptsize}
    \begin{equation*}
        U_V\ket{b^R}_\texttt{Pl} \ket{0}_\texttt{Ut} = \ket{b^R}_\texttt{Pl} \left(\sqrt{1-\frac{V(R)}{V_{\max}-V_{\min}}}\ket{0} + \sqrt{\frac{V(R)}{V_{\max}-V_{\min}}}\ket{1} \right)_\texttt{Ut}.
    \end{equation*}
    \end{scriptsize}
    
    We begin with a quantum state made of three registers: \texttt{Pt}, the partition register, which helps to prepare the $\gamma$ probability amplitude distribution (Definition~\ref{def:shapley}); \texttt{Pl}, the player register, which stores the sub-graph encodings; and \texttt{Ut}, the utility register, which stores the value of a sub-graph.
    We begin with the quantum state,
        $\ket{0}^{\otimes \ell}_\texttt{Pt}
        \ket{0}^{\otimes |F|}_\texttt{Pl}
        \ket{0}^{\otimes 1}_\texttt{Ut}$,
    where $\ell=\mathcal{O}(\log((V_{\max}-V_{\min})\cdot\sqrt{n}/\epsilon))$.
    Next, prepare the \texttt{Pt} register as follows,
    \begin{equation*}
        \frac{1}{\sqrt{2^\ell}} \sum\limits_{k=0}^{2^\ell-1}
        \ket{\upsilon_k}_\texttt{Pt}
        \ket{0}^{\otimes |F|}_\texttt{Pl}
        \ket{0}^{\otimes 1}_\texttt{Ut},
    \end{equation*}
    where $\upsilon_k$ is an $\ell$ bit binary approximation of $\texttt{arcsin}\sqrt{2^{-\ell}k}$.
    For notational simplicity, we suppose $i=|F|-1$.
    Using the partition register as a control, it is efficient to transform the state to,
    \begin{equation}\label{eq:playerPrep}
        \frac{1}{\sqrt{2^\ell}} \sum\limits_{k=0}^{2^\ell-1}
        \ket{\upsilon_k}_\texttt{Pt}
        \left(\left(\sqrt{1-2^{-\ell} k}\ket{0}+\sqrt{2^{-\ell} k}\ket{1}\right)^{\otimes |F|-1}\otimes \ket{0}\right)_\texttt{Pl}
        \ket{0}^{\otimes 1}_\texttt{Ut},
    \end{equation}
    Note that the bit corresponding to node $i$ is zero.
    Switching to a density matrix representation and tracing out the partition register gives an approximation for the state,
    \begin{equation*}
        \sum\limits_{R\subseteq F\setminus \{i\}} \gamma\left(|F\setminus \{i\}|, |R|\right)\ket{b^R}_\texttt{Pl}\ket{0}_\texttt{Ut} \bra{b^R}_\texttt{Pl}\bra{0}_\texttt{Ut}.
    \end{equation*}
    This results from the fact that $\int^1_0 (1-t)^{n-m} t^m dt = \gamma(n,m)$ for integer $n\geq 2$, and $m\in \{0,1,\dots,m\}$.
    Now, applying $U_V$ and measuring the utility bit gives an expected value of,
    \begin{equation} \label{eq:shapNoI}
        \frac{1}{V_{\max}-V_{\min}}\sum\limits_{R\subseteq F\setminus \{i\}} \gamma\left(|F\setminus \{i\}|, |R|\right) V(R).
    \end{equation}
    Using the quantum speedup for Monte Carlo methods~\cite{montanaro2015quantum}, the expected value can be approximated quadratically faster than with classical methods.
    
    We can repeat the process with a simple modification, prepare Equation~\eqref{eq:playerPrep} where the bit corresponding to node $i$ is one, then proceed identically to above.
    This yeilds the expected value,
    \begin{equation} \label{eq:shapYesI}
        \frac{1}{V_{\max}-V_{\min}}\sum\limits_{R\subseteq F\setminus \{i\}} \gamma\left(|F\setminus \{i\}|, |R|\right) V(R\cup \{i\}).
    \end{equation}
    Subtracting Equation~\eqref{eq:shapNoI} from Equation~\eqref{eq:shapYesI}, then multiplying the result by $(V_{\max}-V_{\min})$ gives an approximation for the $i$th player's Shapley value.
    Note that we can compute Equation~\eqref{eq:shapNoI}, Equation~\eqref{eq:shapYesI}, and thus the entire Shapley approximation without measurement.
    As a result, we can approximately perform the transformation,
    \begin{equation}\label{eq:shapOutput}
        \ket{i}\ket{0} \rightarrow \ket{i}\ket{\tilde\Phi_i}.
    \end{equation}
\end{proof}

\begin{lemma}[Shapley values and unreliable value functions]\label{lem:badValueFunction}
    Consider the cooperative game $G_H=(F,V)$ on graph $H=(N,E)$ where $F\subseteq N$.
    We wish to find the Shapley value $\Phi_i$ of node $i$.
    Suppose $V:\mathcal{P}(F)\rightarrow\{0,1\}$ is a binary classifier, and that $V$ is monotonic, if $Q,R\subseteq F$ then $V(Q\cup R) \geq V(Q)$.
    We define $\hat{V}$, which, given $Q\subseteq F$, fails and outputs $1-V(Q)$ with probability $\delta\in[0,1]$, or succeeds and outputs $V(Q)$ with probability $1-\delta$.
    Note, for simplicity, we assume a perfect implementation of the $\gamma$ distribution, in reality, the implementation is an exponentially accurate approximation.
    Applying the Shapley value approximation using $\hat V$ as a substitute for $V$ has expected value
    \begin{equation*}
        \Phi_i+\xi
    \end{equation*}
    where $\xi$ is bounded, $|\xi|\leq 2\delta$.
\end{lemma}
\begin{proof}
    We must find the expected value of the following equation,
    \begin{equation}\label{eq:noisyVshapley}
        \sum_{R\in F\setminus \{i\}} \gamma(|F\setminus\{i\}|,|R|) \left(\hat V(R\cup\{i\})-\hat V(R)\right).
    \end{equation}
    By definition, the expected value of $\hat V(Q)$, $Q\subseteq F$, is $\delta\cdot(1-V(Q))+(1-\delta)\cdot V(Q)$.
    Rearranging gives,
        $\mathbb{E}\left[\hat V(Q) \right] = V(Q) + \delta - 2\delta V(Q)$.
    Thus, Equation~\eqref{eq:noisyVshapley} has expected value,
    \begin{equation*}
        \sum_{R\in F\setminus \{i\}} \gamma(|F\setminus\{i\}|,|R|) \left[(V(R\cup\{i\})-V(R))(1-2\delta) + 2\delta\right].
    \end{equation*}
    Applying Definition~\ref{def:shapley} and Lemma~1 from~\cite{burge2024shapley}, the expected value is equal to,
        $\Phi_i + 2\delta(1-\Phi_i)$.
    Since $V$ is monotonic and outputs in range $\{0,1\}$, then $\Phi_i$ is in range $[0,1]$.
\end{proof}

\subsection{Combining the algorithms}\label{sec:combiningAlgs}

\noindent  In this section, we describe a quantum approach for finding the \emph{st}-connectivity based node centrality.
Consider the cooperative game $G_H=(F,V_{st})$ on graph $H=(N,E)$, where $s,t\in N$ and $F = N \setminus \{s,t\}$.
Suppose we wish to find the Shapley value $\Phi_i$ of node $i\in F$.
We can represent each subset $Q\subseteq F$ with a binary string $b^Q=b_0^Q\cdots b_{|N|-1}^Q$ where $b^Q_j$ is equal to $1$ if $j\in Q$ else $b^Q_j$ is $0$.
Note that, $V_{st}(Q)$ is either $0$ or $1$.
Hence, we can take \hbox{$V_{\max}=1$} and \hbox{$V_{\min}=0$}.

Consider a modified quantum algorithm for \emph{st}-connectivity algorithm based on Remark~\ref{rem:spanSTCentrality}.
We define $U_{st}(Q)$, $Q\subseteq F$ to be the quantum \emph{st}-connectivity algorithm for graph $H_{Q\cup\{s,t\}}$.
This requires a small alteration to the projection $\Pi_x$, Equation~\eqref{eq:Pi_x}. 
We replace $\Pi_x$ with, 
\begin{equation*}
    \Pi_x^Q = \ket{0}\bra{0} + \sum_{(u,w)\in E_Q}\ket{uw}\bra{uw}.
\end{equation*}
This can be done efficiently. 
Instead of directly using the adjacency bit $x_{uw}$, we use the binary value $x_{uw}\land b^Q_u \land b^Q_w$.
Note that this implementation allows us to perform the calculation for all $Q\subseteq F$ in superposition.
The modification makes the algorithm easily compatible with the Shapley value algorithm.

The base quantum algorithm for \emph{st}-connectivity only has a success probability of $9/10$ (Theorem~\ref{thm:quantumSTconnectivity}).
This is insufficient, as is demonstrated in Lemma~\ref{lem:badValueFunction}. 
However, we can improve our accuracy with a logarithmic factor more time and space complexity by repeatedly performing the quantum \emph{st}-connectivity algorithm and taking the majority answer (Lemma~\ref{lem:majVotePower}).
In particular, assuming a desired error $\delta$, we can apply $U_{st}(Q)$ (Remark~\ref{rem:spanSTCentrality}) $n\in\mathcal{O}(\log\delta^{-1})$ times independently and take the majority vote.
We begin with
\begin{equation*}
    U_{st}(Q)^{\otimes n}\bigotimes_{k=0}^{n-1} \ket{0}^{\otimes \mathcal O (\log|N|)}_{\texttt{aux}_k} \ket{0}_{\texttt{out}_k}.
\end{equation*}
Tracing out the auxiliary registers gives us a state of the form required in Lemma~\ref{lem:majVotePower}.
Thus, we can take the majority vote $U_\texttt{MAJ}$ and output it to a new one-qubit register.
If we consider this our utility register \texttt{Ut} described in Theorem~\ref{thm:quantumShap}, we can apply the logic from Lemma~\ref{lem:badValueFunction}.
Specifically, for each the \texttt{Ut} quantum basis vector in the player register \texttt{Pl}, $\ket {b^Q}$, the utility register holds the correct output $V(Q)$ with probability $1-\delta$.
As a result, we can define $U_V$ as the product of repeatedly computing $U_{st}(Q)$ in order $\mathcal O (\log\delta^{-1})$ times, followed by a $U_\texttt{MAJ}$ operation on the outputs.
Thus, by Lemma~\ref{lem:badValueFunction}, the expected value we are extracting, $\Phi_i$, is shifted to $\Phi_i+\xi$, $\xi\leq 2\delta$.
Applying the quantum Monte-Carlo speed-up routine extracts the value $\Phi_i+\epsilon+\xi$.
Since both $\epsilon$ and $\xi$ can be bounded to arbitrarily small values, the algorithm is asymptotically correct.

\subsection{Finding important nodes}\label{sec:findImportantNodes}

\noindent Suppose we wish to find the index of a node with a large Shapley value.
Let node $m$ have the largest Shapley value $\Phi_m$.
We find node $j$ such that their Shapley value $\Phi_j$ is greater than or equal to $\Phi_m-\epsilon$.

\begin{lemma}
    Consider a game $G_H=(F,V)$, where $F$ is a subset of nodes in the graph $H$, and $V:\mathcal{P}(F)\rightarrow \mathbb{R}$ is the value function.
    Suppose player $i$ has the largest Shapley value $\Phi_i$, $\Phi_i\geq\Phi_j$ for all $j\in F$.
    Then, player $i$'s Shapley value has the following lower bound,
    \begin{equation*}
        \Phi_i \geq \frac{V(F)}{|F|}.
    \end{equation*}
\end{lemma}
\begin{proof}
    By the property of efficiency \cite{burge2024shapley}, we have that,
        $\sum_{k=0}^{|F|-1}\Phi_k = V(F)$.
    Suppose that $\Phi_i$ is the maximum Shapley value.
    We proceed by contradiction, let $\Phi_i = (V(F)/|F|) - \epsilon$ for $\epsilon>0$.
    It follows that, for all $k$, $\Phi_k\leq (V(F)/|F|)-\epsilon$.
    Thus,
    \begin{align}
        V(F) = \sum_{k=0}^{|F|-1}\Phi_k \leq \sum_{k=0}^{|F|-1}\left(
        V(F)/|F| - \epsilon \right) = 
        V(F)-|F|\epsilon.
    \end{align}
    A contradiction, thus $\Phi_i$ cannot be less than $V(F)/|F|$.
\end{proof}

As a result, when searching for an important node, at worst, we need precision proportional to $V(F)/|F|$.
Thus, to find our important nodes, we create a uniform superposition of nodes stored in the \texttt{Ind} register, where each is given equal probability,
    $(1/|F|)\sum_{k\in F} \ket{k}_\texttt{Ind}$.
We perform our combined algorithm to assess the Shapley values in the \emph{st}-connectivity game, storing the results $\tilde\Phi_k \approx \Phi_k$ in a new \texttt{Shp} register,
\begin{equation*}
    \frac{1}{|F|}\sum_{k\in F} \ket{k}_\texttt{Ind}\ket{\tilde{\Phi}_k}_\texttt{Shp},
\end{equation*}
where $|\tilde\Phi_k-\Phi_k| \leq \mathcal{O}(V(F)/|F|)$.
We can find the $k$ such that $\tilde\Phi_k$ is maximized in $\mathcal{O}(\sqrt{|F|})$ applications of the combined algorithm using a quantum algorithm for finding the maximum \cite{ahuja1999quantum}.
By excluding players who have already been assessed, this algorithm can be repeated to find multiple high-value players.

\section{Augmented approach with malicious pattern detection}
\label{sec:second-contribution}

\noindent Section~\ref{sec:approach-1} contains a formal approach to identify important nodes in a network that are potential attack victims.  We now narrow our focus to confirm and detect precise patterns associated with malicious actions against a quantum network.
This process can trigger actions to mitigate the impact of malicious activities on identified critical nodes. Hence, the approach presented in this section and that of Section~\ref{sec:approach-1} complement each other. 

\gls*{svm} is a powerful tool historically used for anomaly detection \cite{zhang2006support}.
It has a quantum implementation
based on the least squares formulation of \gls*{svm}s~\cite{rebentrost2014quantum}.
This reduces the \update{training of an SVM to inverting} an matrix and applying it to a label vector.
Suppose we are given an embedding $\Phi: \mathbb{R}^N\rightarrow \mathbb{R}^L$, a data set  $\{x_k\}_{k=1}^M$, $x_k\in \mathbb{R}^N$, and a set of labels $\{y_k\}_{k=1}^M$, $y_k\in\{-1,1\}$.
Then an \gls*{svm} finds an optimal separating hyper plane represented by $\alpha\in \mathbb{R}^L$, and bias $b\in\mathbb{R}$ such that 
\begin{equation}\label{eq:svmTrain}
    \begin{bmatrix}
        0 & \vec{1}^T \\
        \vec{1} & \Omega-\gamma^{-1} I
    \end{bmatrix} ^{-1}
    \begin{bmatrix}
        0 \\ \vec{y}
    \end{bmatrix}
    =
    \begin{bmatrix}
        b\\ \alpha
    \end{bmatrix}.
\end{equation}
where $\vec1$ is the vector of all $1$s, $\vec y$ is the label vector, $\gamma$ is a hyper parameter that manages overfitting, and $\Omega_{i,j}=\Phi(x_i)^T\Phi(x_j)$. 
The classification of a new data point $x$ is determined by the value of $\text{sign}(\sum_{k=1}^M\alpha_k\Phi(x_k)^T\Phi(x)+b)$.

To implement the quantum least squares algorithm, only a few basic goals need to be achieved.
A label vector $y$ needs to be created.
This step is conceptually quite simple.
First, one must generate a uniform superposition. In the case where there are $2^m$ data points, this can be done with $m$ Hadamard gates.
\update{Then, we apply a unitary $U_O$ that flips the phase depending on the category, $U_0\ket k = y_k\ket{k}$.}
The other goal is to apply the inverse matrix of Equation~\eqref{eq:svmTrain}.
In~\cite{rebentrost2014quantum}, this operation is performed using three main techniques, \gls*{hhl}~\cite{harrow2009quantum}, Trotterization~\cite{lloyd1996universal}, and a technique for applying a quantum state as an operation~\cite{lloyd2014quantum}.

Given access to the quantum operation $e^{i\Delta tH}$, where $\Delta t\in \mathbb{R}^+$ and $H$ is a hermitian matrix, \gls*{hhl} performs the operation $H^{-1}$ with repeated applications of $e^{i\Delta tH}$~\cite{harrow2009quantum}.
If $H=\sum_k H_k$, then Trotterization allows one to approximate $e^{i\Delta tH}$ as $\prod_k  e^{i\Delta tH_k}$ with an error dependent on $t$ \cite{rebentrost2014quantum}.
Formally, we have,
\begin{small}\begin{equation*}
    e^{i\Delta t \begin{bmatrix}
        0 & \vec 1^T \\
        \vec 1 & \Omega-\gamma ^{-1}I
    \end{bmatrix}}
    =
    e^{i\Delta t \begin{bmatrix}
        0 & \vec 1^T \\
        \vec 1 & 0
    \end{bmatrix}}
    e^{i\Delta t \begin{bmatrix}
        0 & \vec 0^T \\
        \vec 0 & -\gamma ^{-1}I
    \end{bmatrix}}
    e^{i\Delta t \begin{bmatrix}
        0 & \vec 0^T \\
        \vec 0 & \Omega
    \end{bmatrix}}
    +\mathcal{O}(\Delta t^2).
\end{equation*}\end{small}

Two of the unitaries are relatively trivial to apply. 
The main challenge is reduced to applying a transformation of the form $e^{i\Delta t\Omega}$.
This operation can be performed by leveraging a quantum state with density matrix $\Omega$ and an application of the infinitesimal swap operator $e^{i\Delta tS}$~\cite{lloyd2014quantum}.
Given quantum density matrices $\rho, \Omega$, stored in registers $A, B$ respectively, and swap matrix $S\ket{\psi}_A\otimes\ket\phi_B=\ket\phi_A\otimes\ket\psi_B$, we have,
\begin{equation*}
    \texttt{tr}_B e^{i\Delta t S} \rho \otimes \Omega e^{-i\Delta tS} = e^{i\Delta t \Omega} \rho e^{-i\Delta t \Omega}+\mathcal{O}(\Delta t^2).
\end{equation*}
This is equivalent to applying an approximation for $e^{i\Delta t \Omega}$ to register $A$.
As we need to apply $e^{i\Delta tH}$ repeatedly, we will need to construct many copies of the quantum density matrix $\Omega$.
More details can be found in Section~\ref{subsec:AnomolyDetectComplexity}.

\subsection{Applying synthetic data for training}
\label{subsec:synthetic}

\noindent To construct an arbitrary density matrix $\Omega$, the most common proposed approach is to use quantum RAM~\cite{giovannetti2008quantum}.
However, it is not certain that current proposals for quantum RAM are physically realizable while also allowing for large quantum speedups \cite{Jaques2025qramsurveycritique}.
We consider an alternative approach to generate synthetic data.

Suppose we are given two parameterized quantum state preparation circuits, represented by unitaries $U_0(\theta)$ and $U_1(\theta)$, where
\begin{equation*}
    \theta\in\left\{2^{-\ell}\left(r_1,r_2,\dots,r_{m}\right)^T:r_s\in \mathbb{Z}_{2^\ell}, s=1,\dots,m\right\}. 
\end{equation*}
$\theta$ can be read as an $m$ element vector of $\ell$-bit fixed point integers, or an element of $\mathbb{Z}_{2^{m\ell}}$ when convenient.
Our goal is to use a \gls*{svm} to determine if a state corresponds to the state distribution given by $U_0$ or $U_1$.
We must create a density matrix $\Omega$ that represents the distributions given by $U_0$ and $U_1$.

We define a quantum state with three registers, \texttt{A,B}, which encode the state index, and \texttt{C}, which encodes the state.
Using Hadamard gates, we can construct a quantum state that encodes every possible combination of $\nu\in\{0,1\}$ and $\theta$,
\begin{equation*}
    \frac{1}{\sqrt{2^{m\ell+1}}}\sum_{\nu=0}^1 \sum_{\theta=0}^{2^{m\ell}-1} \ket{\nu}_\texttt{A}\ket{\theta}_\texttt{B}\ket{0}_\texttt{C}.
\end{equation*}

We use register \texttt{A} to control whether $U_0$ or $U_1$ is applied to \texttt{C}, and register \texttt{B} can be used to control the given parameters $\theta$.
Thus, the state is
\begin{equation*}
    \ket{\psi} = \frac{1}{\sqrt{2^{m\ell+1}}}\sum_{\nu=0}^1 \sum_{\theta=0}^{2^{m\ell}-1} \ket{\nu}_\texttt{A}\ket{\theta}_\texttt{B}\ket{x_{k,\theta}}_\texttt{C}
\end{equation*}
where $\ket{x_{\nu,\theta}}=U_\nu(\theta)\ket{0}$.
Discarding register $C$ yields the following density matrix,
\begin{equation*}
    \texttt{tr}_\texttt{C}\ket{\psi}\bra{\psi} = \frac{1}{\sqrt{2^{m\ell+1}}}\sum_{\nu,\iota=0}^1 \sum_{\theta,\eta=0}^{2^{m\ell}-1}
    \braket{x_{\nu,\theta}|x_{\iota,\mu}} \ket{\nu}_\texttt{A}\ket{\theta}_\texttt{B} \bra{\iota}_\texttt{A}\bra{\eta}_\texttt{B}.
\end{equation*}
Each entry in the density matrix corresponds to the similarity between the states $U_k(\theta)\ket0$ and $U_s(\eta)\ket0$.

\subsection{Complex data}
\label{sec:complexData}
\noindent 
Since our data consists of quantum states, it is complex.
In the classical case, complex value \glspl*{svm} have been explored~\cite{bouboulis2014complex}.
However, to the best of our knowledge, the approach described in this section is a novel method for handling complex data in the context of quantum \gls*{svm}s.

It takes some modifications to ensure the similarities between each datapoint are real numbers between $0$ and $1$.
To accomplish this, we craft a kernel matrix $\Omega$ such that $\Omega_{s,k}=\left|\braket{x_s|x_k}\right|^2$.
We define $\bar{U}_\nu(\theta)\ket{0}=\overline{\ket{x_{\nu,\theta}}}$.
$\bar{U}_\nu(\theta)$, 
constructed by conjugating each elementary operation composing $U_\nu(\theta)$.
Consider a quantum state with four registers, \texttt{A,B}, which encode the state index, and $\texttt{C}_0,\texttt{C}_1$, which encodes the state.
Using the same method as in Section~\ref{subsec:synthetic}, we construct the state
\begin{equation*}
    \frac{1}{\sqrt{2^{m\ell+1}}}\sum_{\nu=0}^1 \sum_{\theta=0}^{2^{m\ell}-1} \ket{\nu}_\texttt{A}\ket{\theta}_\texttt{B}\ket{0}_{\texttt{C}_0}\ket{0}_{\texttt{C}_1}.
\end{equation*}

Using registers \texttt{A,B} as controls, we apply $U_\nu(\theta)$ to register $\texttt{C}_0$ and $\bar{U}_\nu(\theta)$ to register $\texttt{C}_1$.
This yields the state
\begin{equation*}
    \ket{\psi} = \frac{1}{\sqrt{2^{m\ell+1}}}\sum_{\nu=0}^1 \sum_{\theta=0}^{2^{m\ell}-1} \ket{\nu}_\texttt{A}\ket{\theta}_\texttt{B}\ket{x_{\nu,k}}_{\texttt{C}_0}\overline{\ket{x_{\nu,k}}}_{\texttt{C}_1}.
\end{equation*}
Then, tracing out $\texttt{C}_0,\texttt{C}_1$, results in the density matrix
\begin{equation*}
    \texttt{tr}_{\texttt{C}_0,\texttt{C}_1}\ket{\psi}\bra{\psi} = \frac{1}{\sqrt{2^{m\ell+1}}}\sum_{\nu,\iota=0}^1 \sum_{\theta,\eta=0}^{2^{m\ell}-1}
    \left|\braket{x_{\nu,\theta}|x_{\iota,\eta}}\right|^2 \ket{\nu}_\texttt{A}\ket{\theta}_\texttt{B} \bra{\iota}_\texttt{A}\bra{\eta}_\texttt{B}.
\end{equation*}

Suppose that we wish to classify a state $\ket{x}$.
We describe a modified version of the classification procedure from Rebentrost \emph{et al.} \cite{rebentrost2014quantum}.
Define normalizing factors $\mathcal{N}_\mu=b^2+ \sum_{b,\theta}\alpha_{b,\theta}^2$ and $\mathcal{N}_x=2^{m\ell+1}+1$.
We begin with the hyperplane normal vector state,
\begin{equation*}
    \ket{\mu} = \frac{1}{\sqrt{\mathcal{N}_\mu}}\left(b\ket{0}\ket0\ket0+ \sum_{\nu=0}^1 \sum_{\theta=0}^{2^{m\ell}-1}\alpha_{\nu,\theta}\ket{\nu}\ket{\theta}\ket{x_{\nu,\theta}}\right),
\end{equation*}
and a query state
\begin{equation*}
    \ket{\tilde x}=\frac{1}{\sqrt{\mathcal{N}_x}}\left(\ket{0}\ket{0}\ket0 + \sum_{\nu=0}^1 \sum_{\theta=0}^{2^{m\ell}-1} \ket{\nu}\ket{\theta}\ket{x}\right).
\end{equation*}

We apply the swap test (Figure~\ref{fig:swaptest}) to the $(2n+1)$-qubit state $\ket{0}\ket{\mu}\ket{\tilde x}$, yielding the state 
\begin{equation*}
    \frac{1}{2}\ket{0}\left(\ket{\mu}\ket{\tilde x} + \ket{\tilde x}\ket\mu\right) +
    \frac{1}{2}\ket{1}\left(\ket{\mu}\ket{\tilde x} - \ket{\tilde x}\ket\mu\right)
\end{equation*}
The probability of measuring $0$ in the first bit is
\begin{equation*}
    \frac{1}{4}\left(\bra{\mu}\bra{\tilde x} + \bra{\tilde x}\bra\mu\right)
    \left(\ket{\mu}\ket{\tilde x} + \ket{\tilde x}\ket\mu\right),
\end{equation*}
which evaluates to,
\begin{equation} \label{eq:classification}
    \frac{1}{2}+\frac{1}{2\sqrt{\mathcal{N}_\mu\mathcal{N}_x}}\left[ b + \sum_{k=1}^{M}\alpha_k\left| \braket{x_k|x}\right|^2 \right].
\end{equation}
We classify $x$ as $-1$ if the probability of measuring $0$ is greater than $0.5$, otherwise we classify $x$ as $1$.

\begin{figure}
    \centering
    \begin{quantikz}
        \lstick{$\ket{0}$} & \gate{H} & \ctrl{2} & \gate{H} & \meter{} \\
        \lstick{$\ket{\mu}$} & & \targX{} & &\\
        \lstick{$\ket{\tilde x}$} & & \targX{} & &\\
    \end{quantikz}
    \caption{Swap test to assess class of datapoint $x$~\cite{buhrman2001quantum}.}
    \label{fig:swaptest}
\end{figure}

\section{Practical example}
\label{sec:experiments}

\noindent We consider an explicitly defined quantum network topology.
The section is outlined as follows: Sections~\ref{sec:practicalExampleSTcon} and~\ref{sec:complexityOfSTShapley} deal with \update{finding the \emph{st}-connectivity-based} node centrality for our scenario.
After identifying important nodes based on our metric, we discuss attack detection.
Sections~\ref{sec:detectingEntAttack} and~\ref{subsec:AnomolyDetectComplexity} explore our entanglement detection strategy for corrupted nodes.

\begin{figure}[!t]
    \centering
    \begin{subfigure}[t]{\columnwidth}
    \includegraphics[width=\columnwidth]{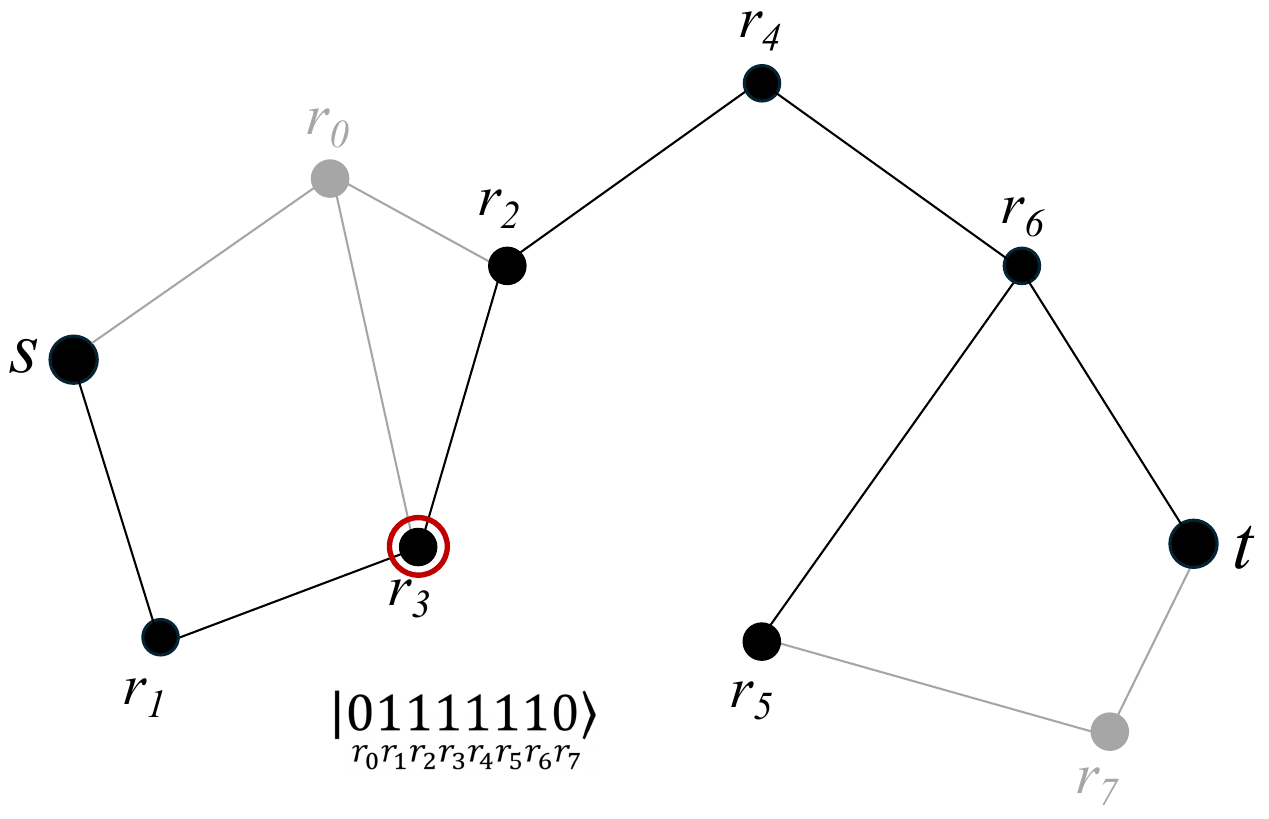}
    \caption{Binary representation of a subgraph.\label{fig:coaltion_example}}
    \end{subfigure}
    \begin{subfigure}[t]{\columnwidth}
        \includegraphics[width=\columnwidth]{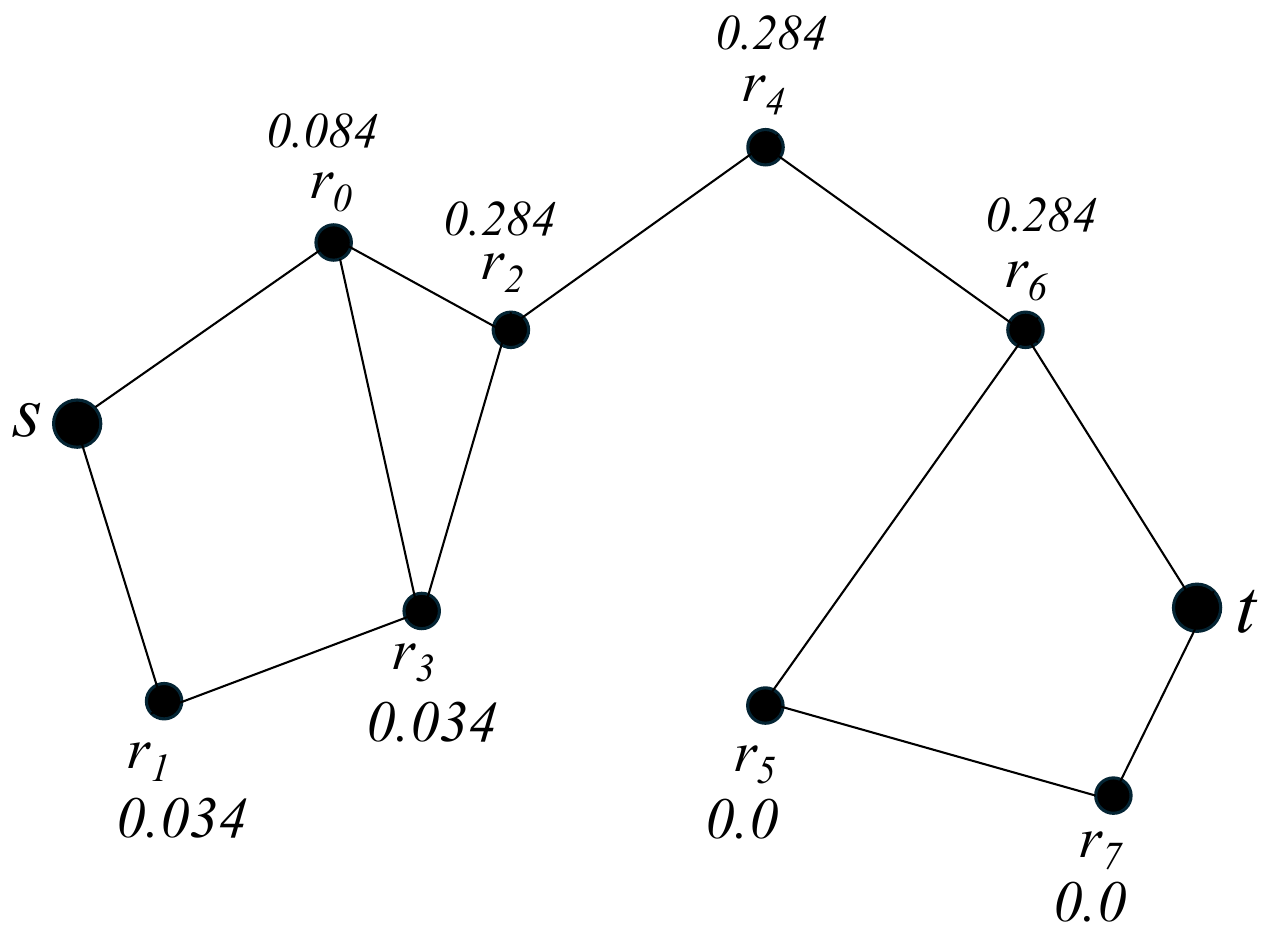}
    \caption{All intermediary nodes labeled with their corresponding Shapley Value.\label{fig:complete_example}}
    \end{subfigure}        
    \caption{Practical example. (a) Binary representation of subgraph where $r_3$ decides if nodes $s$ and $t$ are connected. (b) Complete example with all the intermediary nodes labeled with their corresponding Shapley Value (cf. our companion repository~\cite{github} for further details).}\label{fig:coalitions}
\end{figure}

\subsection{Assessment of critical nodes using Shapley values}
\label{sec:practicalExampleSTcon}
 
\noindent 
Let $H=(N,E)$ be the graph as shown in Figure~\ref{subfig:basicNetwork}.
We define a \update{cooperative game $G_H = (F, V_{st})$, with} $s,t\in N$ and $F=N\setminus \{s,t\}$.
Suppose we wish to find the Shapley value $\Phi_3$ of node $r_3\in F$.
We represent each subset $R\subseteq F$ with the binary string $r^R=r_0^Rr_1^Rr_2^Rr_3^Rr_4^Rr_5^Rr_6^R r_7^R$ where $r^R_j$ is equal to $1$ if $j\in R$ else $r^R_j$ is $0$.
\update{Note that, $V_{st}(R)$ is} either $0$ or $1$. Hence, we can take \hbox{$V_{\max}=1$} and \hbox{$V_{\min}=0$}. We define $U_V$ as done in the previous section.
For example, suppose we apply $U_V$ with input string $\ket{0111110}_\texttt{Pl}$, which represents the coalition depicted in Figure~\ref{fig:coaltion_example}.
Clearly, this subset is \emph{st}-connected, since the path $s,r_1,r_3,r_2,r_4,r_6,t$ is valid.
As a result, if we perform $U_V\ket{0111110}_\texttt{Pl}\ket{0}_\texttt{Ut}$, the \texttt{Ut} register stores the correct answer, $1$, with probability $1-\delta$. 
However, if we remove node $r_3$, the graph is no longer \emph{st}-connected, i.e., the state $U_V\ket{0110110}_\texttt{Pl}\ket{0}_\texttt{Ut}$ has the answer $0$ stored in the \texttt{Ut} register with probability $1-\delta$. 

To find Shapley value $\Phi_3$, we proceed as follows:
(i) craft a quantum state that encodes every possible subset of nodes, 
that does not include node $r_3$, with correct amplitude probability weights corresponding to $\gamma$; 
(ii) perform the unitary $U_V$ outputting to \texttt{Ut}, i.e., repeatedly check for \emph{st}-connectivity leveraging Theorem~\ref{thm:quantumSTconnectivity} and take the majority answer;
(iii) extract the expected value of the utility register \texttt{Ut} using the Monte-Carlo speed-up \cite{montanaro2015quantum};
(iv) repeat the previous steps where each subset that includes 
node $r_3$ is considered and compare outputs.
Using this strategy, we can approximate the Shapley values of each node to arbitrary accuracy (cf. Figure~\ref{fig:complete_example}).
As a result, we can also leverage the techniques described in Section~\ref{sec:findImportantNodes}, to quickly identify which nodes have the highest Shapley values, i.e., nodes that represent valuable targets for a potential attack.

\subsection{Complexity analysis of the assessment approach}
\label{sec:complexityOfSTShapley}

\noindent \emph{Baseline Classical Complexity --}
We now describe a reasonable, though not necessarily optimal method to \update{approximate \emph{st}-connectivity-based node centrality} through classical methods.
Let $G_H=(F,V_{st})$ be a cooperative game on $H=(N,E)$, where $s,t\in N$ and $F=N\setminus\{s,t\}$.
Let us discuss the complexity of approximating player $i$'s Shapley value.
The \emph{st}-connectivity can be assessed using breadth first search, with a time complexity of $\mathcal{O}(|N|^2)$.
By Chebyshev's inequality, we need to query the \emph{st}-connectivity algorithm $\mathcal{O} (\sigma^2/\epsilon^2)$ times, where $\epsilon$ is the desired error, and $\sigma^2$ is the variance of $V_{st}$ over the distribution matching the Shapley value Definition~\ref{def:shapley}.
Since the only outputs of $V_{st}$ are zero and one, we effectively have a Bernoulli distribution with expected value $\Phi_i$.
Thus, the variance is $\Phi_i(1-\Phi_i)$.
Since non-trivial situations do not allow for $\Phi_i$ to be close to one, we effectively have a variance of $\mathcal{O}(\Phi_i)$.
Thus, given a fixed likelihood of success, the time complexity of approximating the Shapley value $\Phi_i$ with error bounded by $\epsilon$ is
    $\mathcal{O} \left({\Phi_i}{\epsilon^{-2}} |N|^2\right)$.

Next, we briefly consider a method to extract important nodes.
In the worst case, the largest Shapley value is of size $\mathcal O(V(F)/|F|) = \mathcal{O} (1/|N|)$. \update{In this case, most values are close together.} An error bound $\epsilon\in \mathcal O(1/|N|)$ and Shapley value $\Phi_i\in\mathcal O (1/|N|)$ are appropriate values.
Thus, we require $\mathcal O(|N|^3)$ operations for sufficient accuracy.
Finally, we must find the Shapley value \update{for each node; thus, naively, the worst-case scenario} involves about $\mathcal O(|N|^4)$ operations.

\emph{Quantum Complexity --}
Let us now address the complexity of our quantum approach.
Note that we drop logarithmic factors for notational simplicity.
We describe the complexity of approximating player $i$'s Shapley value with quantum methods.
$U_V$ involves repeating the algorithm from Theorem~\ref{thm:quantumSTconnectivity} a logarithmic number of times.
Thus, $U_V$ has a time complexity of $\tilde {\mathcal O}\left(|N|^{3/2}\right)$.
Note that Theorem~\ref{thm:quantumSTconnectivity} implicitly requires an easily addressable form of adjacency matrix. 
In this context, the Shapley value algorithm has complexity $\tilde {\mathcal O} \left(\sqrt{\Phi_i}/\epsilon\right)$ (Theorem~\ref{thm:quantumShap}).
Thus, the complexity for finding node $i$'s Shapley value is $\tilde {\mathcal{O}} \left(\sqrt{\Phi_i} \epsilon^{-1} |N|^{3/2}\right)$.

Applying \update{the same reasoning} as above, we consider the problem of extracting important nodes.
\update{Suppose that the largest} Shapley value is $\mathcal O(1/|N|)$ \update{and as a result we target} $\epsilon\in \mathcal{O}(1/|N|)$.
Thus, to compute Shapley values to the required precision takes $\tilde {\mathcal O}(|N|^2)$ time.
As discussed in Section~\ref{sec:findImportantNodes}, we can approximate all Shapley values in superposition, then extract the maximum in $\tilde {\mathcal{O}} \left(\sqrt{|N|}\right)$ queries.
Thus, our total complexity for finding important nodes takes $\tilde {\mathcal{O}}(|N|^{5/2})$ operations.

\subsection{Detecting an entanglement attack}
\label{sec:detectingEntAttack}
\noindent Figure~\ref{fig:coalitions} shows that nodes $r_2,r_4,$ and $r_6$ are valuable targets. Based on the network topology, $r_2$ and $r_6$ are routers, while $r_4$ could be implemented as a repeater.
As a result, $r_4$ would be a likely target for an attacker.
Recall the threat model from Section~\ref{sec:motivation-part1}, we describe a method to determine whether $r_4$ is compromised and performing an entanglement attack.

To detect an entanglement attack originating from $r_4$, we perform the following steps.
Node $r_2$ constructs two identical \update{two-qubit quantum states denoted as} $A_1,B_1$ for the first pair, and $A_2,B_2$ for the second pair, where $A_k$ and $B_k$ are entangled.
$r_2$ sends both pairs, $A_1,B_1$ and $A_2,B_2$, to $r_6$ via $r_4$.
Based on our adversarial model, if $r_4$ is compromised, it is possible for $r_6$ to receive,
$A_1B_1A_2B_2$,
$A_1C_1A_2B_2$,
$A_1B_1A_2C_2$, or
$A_1C_1A_2C_2$, where $C_k$ is an arbitrary qubit not entangled with $A_k$.
Finally, $r_6$ uses a quantum \gls*{svm} trained with synthetic data that distinguishes the expected state $A_1B_1A_2B_2$, from malicious states $A_1C_1A_2B_2$, $A_1B_1A_2C_2$, or $A_1C_1A_2C_2$.

Consider the parameterized circuit design for creating entangled pairs (Figure~\ref{fig:synthCircuit}) described in Mahdian \emph{et al.}~\cite{mahdian2025entanglement}.
Proceeding as described in Section~\ref{subsec:synthetic}, we denote the circuit in Figure~\ref{fig:preparation-circuit-a} with the unitary $V(\theta)$ and Figure~\ref{fig:preparation-circuit-b} with $W(\theta)$.
Define the following parameterized unitaries $U_0(\theta),U_1(\theta)\in\mathbb{C}^{2^4 \times 2^4}$ as,
\begin{equation*}
    U_0(\theta) = V(\theta)\otimes V(\theta),
\end{equation*}
\begin{equation*}
    U_1(\theta) = 
    \begin{cases}
        V(\theta)\otimes W(\theta') & \text{if } \theta\text{ mod } 5 \in\{0,1\},\\
        W(\theta)\otimes V(\theta) & \text{if }\theta\text{ mod } 5 \in\{2,3\},\\
        W(\theta)\otimes W(\theta') & \text{if }\theta\text{ mod } 5 \in \{4\}. \\
    \end{cases}
\end{equation*}
where $\theta'$ is a pseudo-random vector derived from seed $\theta$. 
For example, let $\theta'_k=z_k\theta_k \mod 2^\ell$, $z_k\in\mathbb{Z}$, where $z_k$ is co-prime to $z_s$ if $k\neq s$.
The definition of $U_1(\theta)$ is intended to have substantial data on each possible attack pattern.
\update{In principle, the definition of $U_1(\theta)$ could be tuned to better represent the distribution of attacks or to reduce the computational complexity of training.}

Leveraging the techniques of Section~\ref{sec:complexData}, we construct our kernel and train our desired quantum \gls*{svm}.
Since the trained quantum \gls*{svm} is a quantum state, training can be performed on a more capable machine, the result can be sent to node $r_6$.
$r_6$ performs the modified swap test to categorize the state as expected or malicious.
A numerical implementation of our approach, including code and artifacts, is available in a companion repository~\cite{github}. Empirically, the resulting quantum \gls*{svm} has a high accuracy (cf. Figure~\ref{fig:confusionMatrix}). Next, we explore the cost of extracting the results.

\begin{figure}
    \centering
    \begin{subfigure}[t]{.85\columnwidth} 
        \centering
        \begin{quantikz}
            \lstick{$\ket{0}$} & \gate{R_x(\theta_0)} & \ctrl{1} & \gate{R_y(\theta_2)} &\\
            \lstick{$\ket{0}$} & & \gate{R_y(\theta_1)} & \gate{R_x(\theta_3)} &
        \end{quantikz}
        \caption{Entangled preparation circuit.\label{fig:preparation-circuit-a}}
    \end{subfigure}
    \begin{subfigure}[t]{.85\columnwidth}
        \centering
        \begin{quantikz}
            \lstick{$\ket{0}$} & \gate{R_x(\theta_0)} & \gate{R_y(\theta_2)} &\\
            \lstick{$\ket{0}$} & \gate{R_y(\theta_1)} & \gate{R_x(\theta_3)} &
        \end{quantikz}
        \caption{Unentangled preparation circuit.\label{fig:preparation-circuit-b}}
    \end{subfigure}
    
    \caption{
        (a) Entangled preparation circuit and (b) unentangled preparation circuit, where $\theta\in\mathbb{R}^4$, $R_x(\omega)=[[\cos\omega,-i\sin\omega],[-i\sin\omega,\cos\omega]]$, and $R_y(\omega)=[[\cos\omega,-\sin\omega],[\sin\omega,\cos\omega]]$.
    }
    \label{fig:synthCircuit}
\end{figure}

\subsection{Complexity analysis of the detection approach}
\label{subsec:AnomolyDetectComplexity}
\noindent
According to the analysis from Rebentrost \emph{et al.}~\cite{rebentrost2014quantum}, assuming we use synthetic data, quantum \gls*{svm} training has a computational complexity of $\mathcal{O}\left(\kappa^3\epsilon^3C\right)$.
$\kappa=\lambda_\text{max}\lambda_\text{min}^{-1}$ is the condition number, or the ratio of the largest to smallest eigenvalues of our synthetic data kernel matrix, $\epsilon$ is a bound on error, and $C$ is the complexity of implementing the parameterized circuits $U_0$,$U_1$.
This is the result of constructing $\text{Poly}(\kappa,\epsilon)$ copies of $\Omega$ and leveraging them for the \gls*{hhl} subroutine.
Note that the complexity of these operations are an active area of research, for example, in some cases, the \gls*{hhl} subroutine has had exponential improvements in terms of complexity with respect to $\epsilon$ \cite{childs2017quantum}.
For simplicity, we continue with the original complexity analysis.
One of the primary limitings factor of quantum \gls*{svm}s are their condition numbers $\kappa$, though it may be possible to design synthetic data such that $\kappa$ remains small even with large datasets.

\begin{figure}[!b]
\centering
    \includegraphics[width=\columnwidth]{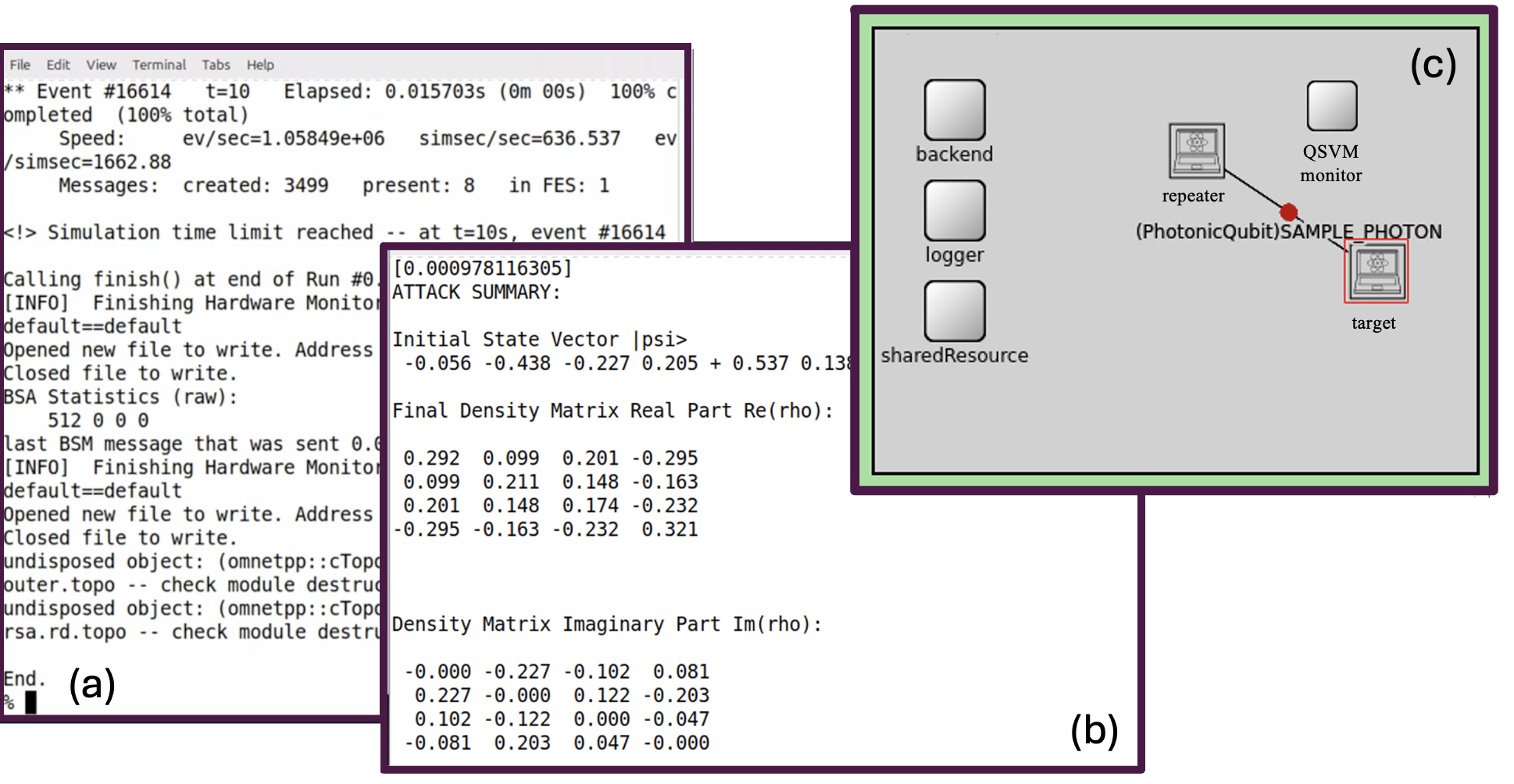}
    \caption{\update{OMNET++-based Entanglement Attack Simulation implemented under the QuISP framework. All the simulation code  and resulting artifacts are released in a companion repository~\cite{github}. (a,b) OMNeT++ Discrete Event Simulation console~\cite{varga2008overview} and attack reporting console, displaying teleported state-vectors and density matrices. (c) Graphical User Interface, using the QuISP framework~\cite{satoh2022quisp}. \label{fig:quisp}}}
\end{figure}

Performing the classification steps from Section~\ref{sec:complexData} is computationally feasible. 
Recall Equation~\eqref{eq:classification}, in order to classify a state $\ket{x}$ with high confidence, we must determine if the expected value of the measurement is above or below $0.5$.
Let $f(x)=\left(\mathcal{N}_x\mathcal{N}_\mu\right)^{-\frac{1}{2}}\left(b+\sum_{\eta,\theta}\alpha_{\eta,\theta}|\braket{x_{\eta,\theta}|x}|^2\right)$ be proportional to the right hand term. In the usual case that $\alpha$ is non-sparse, it is has been implied that $f(x)$ is $\mathcal{O}(1)$~\cite{rebentrost2014quantum}.
This complexity would imply that it is easy to distinguish categories as the number of datapoints grow.
However, in our empirical testing, the right term is made small by the normalizing factors, meaning our expected value of $f(x)+1/2$ is very close to $0.5$.

\update{To experimentally validate the theoretical work, we ran our detection method on a series of datasets corresponding to the target attack described in Section~\ref{sec:detectingEntAttack}.
The attack is implemented using OMNeT++ 6.0.3 ~\cite{varga2008overview} and QuISP 0.3~\cite{satoh2022quisp}. All the simulation code  and resulting artifacts are released in our companion repository~\cite{github} (see Figure~\ref{fig:quisp} for a representative screenshot with the experimental environment). To evaluate the performance of the detection model, we produced $512$ legitimate scenarios and $512$ malicious scenarios.
In the legitimate scenarios, $A_1B_1A_2B_2$ were received.
In the malicious scenarios, $A_1C_1A_2B_2$, $A_1B_1A_2C_2$, or $A_1C_1A_2C_2$, were received with probabilities $40\%$, $40\%$ or $20\%$ respectively. In particular, with $\ell=2$ ($512$ datapoint training set), we found the average of $|f|$ over the trials was $1.36\times10^{-3}$ (std=$8.83\times10^{-4}$).
But, with $\ell=3$ ($8192$ datapoint training set), we found $|f|$ to average $2.22\times10^{-4}$ (std=$1.83\times10^{-4}$). Figure~\ref{fig:confusionMatrix} provides a confusion matrix summarizing the obtained results. It shows that the detection model makes the correct classification with high accuracy, in practice.}

\begin{figure}[!hptb]
    \centering
    \includegraphics[width=.9\linewidth]{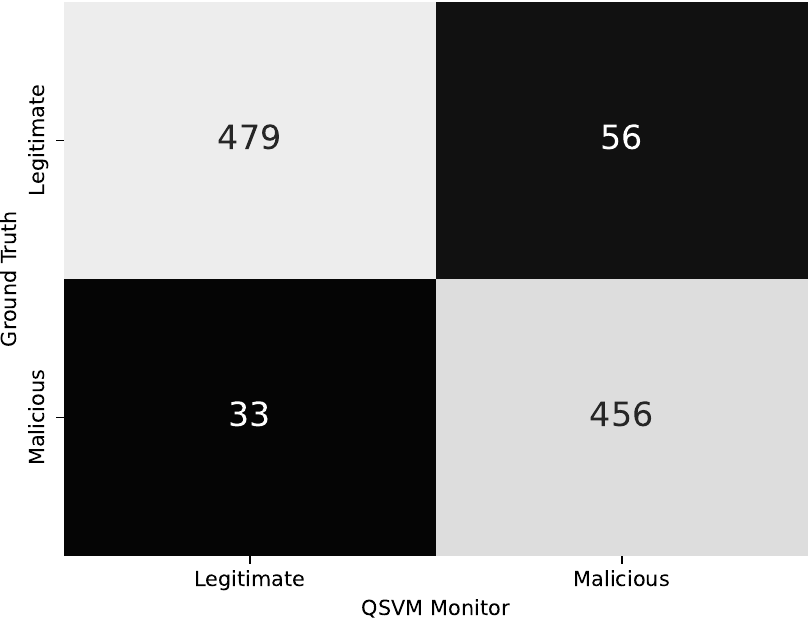}
    \caption{\update{Experimental results obtained with our quantum \gls*{svm} monitor, trained with the technique described in Sections~\ref{subsec:synthetic} and~\ref{sec:complexData} with $m=4$, $\ell=2$, $\gamma=1$, and circuits described in Figure~\ref{fig:synthCircuit} ($512$ datapoint training set).
    The confusion matrix shows the classifications of a balanced dataset generated using an attack simulated in QuISP~\cite{satoh2022quisp,chung2025cross}. The top left quadrant represents true valid state classification; the bottom left quadrant is a false valid state classification; the top right is false malicious state classification; and the bottom right represents a true malicious state classification.
    All artifacts, including datasets and detection code to reproduce these results are available at in a companion repository~\cite{github}.}}
    \label{fig:confusionMatrix}
\end{figure}

\section{Related Work}
\label{sec:merged-RW}

\subsection{Quantum assessment of critical nodes}

\noindent The work presented in this paper combines quantum computing with distributed system security. Some existing research directions related to our work include (i)~the study of potential \update{advantage or speed-up
optimizations} of quantum computing \update{associated with} probing, control, and
planning of cyber-physical systems~\cite{barbeau2022}, as well as
formally verifying properties and providing explainability of the
related processes~\cite{chareton2021}; (ii)~use of quantum
technologies to secure quantum data communications (e.g., protecting
the authenticity of quantum signals when in transit, detection of
adversaries maliciously modifying quantum messages, and analysis of
any other threat models affecting the security of entanglement rates
to endanger applications built upon distributed quantum
networks~\cite{barbeau-perez2022}); (iii)~advantages of quantum
technologies to build secure ways to protect classical data with
key expansion protocols like \gls*{qkd}, any of
its flavors~\cite{noirie2024}; (iv)~risks and threats posed by quantum
science to contemporary information security, including the use of
quantum annealers or any other quantum-inspired metaheuristics paving
the way for new cracking strategies \update{against classical or quantum-safe}
cryptography~\cite{chen2025quantum}.

Compared to previous work, we provide in this article a formal approach built upon \update{game-theoretic node centrality} following
in line with~\cite{tarkowski2017game,michalak2013efficient}.
\update{Game-theoretic node} centrality provides a flexible and nuanced concept of node centrality. 
The \emph{st}-connectivity attack, in the \update{context of game-theoretic} node centrality, relies on novel methods to quantify the security properties of a graph. As previously shown~\cite{burge2023quantum,burge@ifipsec2024}, the Shapley values necessary for our node centrality can be approximated with quadratically fewer value function queries using quantum methods, up to polylogarithmic factors. 
Simultaneously, our value function, based on \emph{st}-connectivity, can
be assessed faster on a \update{quantum computer, leveraging the work of Beloys and Reichardt} ~\cite{belovs2012span}. 
The combination of these two factors allows for a faster calculation than is possible with a classical Monte Carlo approach to solving the problem. 
Finally, to find high-importance nodes, we can calculate each node's Shapley value simultaneously using superposition, which yields a database of Shapley values.
We can \update{search this node database to find} the node with \update{the largest Shapley value} quadratically faster than a standard search would allow~\cite{ahuja1999quantum}.

\subsection{Quantum support vector machines for pattern detection}

\noindent \update{Quantum machine learning is a  primary axis of research~\cite{biamonte2017quantum}. The goal is speeding up difficult machine learning problems.} In the fault-tolerant context, there are multiple interesting directions, including quantum~\gls*{svm}s~\cite{rebentrost2014quantum} and quantum neural networks~\cite{jeswal2019recent}.
For \update{near-term hardware,} there are also quantum approaches to problems including support vector machines, linear regression, and balanced $k$-means clustering~\cite{date2021qubo}.
\update{In addition, supervised, unsupervised, and reinforcement learning have been explored~\cite{skolik2022quantum, rapp2025reinforcement}.}

In their initial conception~\cite{rebentrost2014quantum}, quantum \gls*{svm}s leverage multiple sophisticated \update{subroutines and quantum RAM} to perform supervised learning. 
Two of these requirements present issues. 
First, the \update{matrix inversion} subroutine of \gls*{hhl} requires a well-structured kernel matrix, which depends 
\update{on the dataset and embedding~\cite{harrow2009quantum}. 
Second, quantum RAM is a controversial tool. It may not be possible to implement effectively, potentially diminishing or canceling the quantum SVM speedup~\cite{Jaques2025qramsurveycritique}. 
To avoid the issue, synthetic data may be a useful avenue to explore, as they can be more intentionally structured.
The latter issue is solved outright by using synthetic data, since quantum RAM is no longer required to load the dataset.}

Synthetic data has already been leveraged in several instances for quantum machine learning.
In Mahdian \emph{et al.}~\cite{mahdian2025entanglement}, a classical-quantum hybrid \gls*{svm} is used for entanglement detection.  \update{They target near-term hardware, making implementation feasible, at the cost of losing the potential for a quantum speedup in training. In the context of fault-tolerant quantum computing, rigorous speedups have been established that require specially constructed quantum datasets~\cite{liu2021rigorous}.}

\update{Our new detection method provides new boundaries for quantum information security when making a binary decision between two quantum states. This has been reported as an open problem in recent literature (e.g., Lloyd and Xie \cite{lloyd2026risk}). The solution is consistent with current methodologies and architectures for operational security in future quantum entanglement networks (e.g., Chung et al.~\cite{chung2026interqnet}). Given any quantum state classification problem, if parameterized quantum circuits are available that can generate each category, it is possible to leverage \gls*{qsvm}s to perform categorization. The two primary differences to related work are as follows: (i)~we leverage massive synthetic datasets through a novel use of parameterized circuits; and (ii)~the method can deal with complex data.
Taken together, these enable the application of \gls*{qsvm}s to distinguish certain quantum distributions without the need for quantum RAM.
Assuming one can find well-behaved parameterized circuits, it may also be possible to reduce the complexity associated with the condition number in the \gls*{hhl} step of the training.}

\section{Conclusion}
\label{sec:conc}

\noindent \update{The first contribution described in this article is a quantum approach for approximating the importance of nodes that maintain a target connection. The approach identifies high-importance nodes, likely to become victims of potential attacks. The contribution builds upon multiple subroutines: one for \emph{st}-connectivity assessment, another for Shapley value approximation, and a final subroutine for finding the maximum value in a list.}

\update{This first contribution also provided a formal attack scheme, denoted the \emph{st}-connectivity, that serves as the main scenario used for evaluation purposes. It is assumed that a malicious actor disrupts a subset of nodes with the goal of altering the functionality of the system. Based on this first contribution, an automated approach could be proposed to identify the nodes that are the most important and use this information to guide topological reorganization to increase resilience.} 

\update{Building upon the first as a security metric to guide defense strategies, e.g., threat detection, 
the second contribution enhances the assessment of node importance and uses \gls*{qsvm} classifiers to identify malicious events, according to our threat model. A specific method has been evaluated for identifying and detecting the use of malicious entanglement.}

\update{Our complexity analyses compare our quantum algorithms against a straightforward classical baseline algorithm. Further work is necessary to compare with more optimized classical graph algorithms, including empirical runtime comparisons. Hence, a potential practical quantum advantage can be further brought out. Other perspectives for future work include exploring the notion of distributed quantum algorithms, with the goal of identifying situations that can provide enough advantage to either player to break ties in unbalanced security-based designs~\cite{cachin2021asymmetric,barbeau2022}. 
The security of quantum distributed systems, by itself, presents ongoing challenges worth exploring.}
\newpage

\bibliographystyle{elsarticle-num}

\end{document}